\documentclass[journal,onecolumn,11pt,draftcls]{IEEEtran}
\usepackage[T1]{fontenc}
\usepackage{amsmath,amssymb,bm}
\usepackage{cite}
\usepackage[nohyperlinks]{acronym}
\usepackage{graphicx,xcolor,tikz}
\usepackage{needspace}
\usetikzlibrary{arrows.meta,positioning,calc,decorations.pathreplacing}
\usepackage[hidelinks]{hyperref}
\allowdisplaybreaks[1]
\newcommand{\Set}[1]{\mathcal{#1}}
\newtheorem{theorem}{Theorem}
\newtheorem{proposition}{Proposition}
\newtheorem{lemma}{Lemma}
\newtheorem{corollary}{Corollary}

\newacro{ISAC}{integrated sensing and communications}
\newacro{OFDM}{orthogonal frequency division multiplexing}
\newacro{IDFT}{inverse discrete Fourier transform}
\newacro{CP}{cyclic prefix}
\newacro{ESL}{expected sidelobe level}
\newacroplural{ESL}[ESLs]{expected sidelobe levels}
\newacro{EISL}{expected integrated sidelobe level}
\newacro{QAM}{quadrature amplitude modulation}
\newacro{PSK}{phase shift keying}
\newacro{FFT}{fast Fourier transform}

\newacro{ISL}{integrated sidelobe level}
\newacro{PSL}{peak sidelobe level}
\newacro{IEEE}{Institute of Electrical and Electronics Engineers}
\newacro{IRE}{Institute of Radio Engineers}
\newacro{VTC}{Vehicular Technology Conference}
\newacro{6G}{sixth-generation mobile communications}

\begin{document}
\title{\huge On the Optimality of \acs*{OFDM} for Aperiodic Ranging}
\author{Yifeng Xiong, Fan Liu, Ya-Feng Liu, Shi Jin, and Jianhua Zhang}
\maketitle
\begin{abstract}
\begingroup
We study the expected aperiodic autocorrelation sidelobes of random communication symbols transmitted through a unitary modulation matrix. We consider proper sub-Gaussian constellations, namely those with fourth absolute moment \(\mu_4<2\) and zero pseudo-variance, including the commonly used uniform square \ac{QAM} and uniform \ac{PSK} constellations. {For every prescribed lag, we determine the exact minimum \ac{ESL} and a unitary modulation attaining it by applying a normalized \ac{IDFT} separately on each chain of sample indices spaced by that lag. The full \ac{OFDM} modulation attains this minimum if and only if the lag divides the block length \(N\).} Combining this criterion with the complete characterization of the minimizers at lag one, we show that a unitary modulation minimizes the \ac{ESL} simultaneously over the first \(r\) positive lags, called the prefix \(\{1,\ldots,r\}\) with \(1\le r<N\), if and only if every lag in that prefix divides \(N\). In particular, simultaneous optimality over all positive lags is impossible for \(N\ge3\). We also prove that \ac{OFDM} globally minimizes the partial \ac{EISL}, defined as the sum of the \acp{ESL} over a selected subset of positive lags, for every such prefix. This also resolves the full \ac{EISL} conjecture when \(r=N-1\).
\par\endgroup
\end{abstract}

\acresetall
\section{Introduction}\label{sec:tit-introduction}

\subsection{Background and Research Questions}

Communication signals convey information to their intended receivers while their reflections reveal properties of the surrounding environment. \Ac{ISAC} uses a common transmitted waveform for both tasks, making the sensing behavior of communication signals an important part of waveform design \cite{liu2022isac}. As the data change across transmissions, a fixed modulation scheme generates a family of waveforms whose statistical properties affect the observations available to the sensing receiver \cite{liu2026communicationsignals}. The signaling distribution determines the tradeoffs between communication rate and estimation performance \cite{xiong2023tradeoff,ahmadipour2024joint}, providing a broader motivation for understanding how random information symbols should be mapped to a transmitted signal.

Communication waveforms must convey information while providing a signal structure that allows the receiver to distinguish and recover different symbols. Orthogonal modulation assigns the symbols to mutually orthogonal basis functions, providing a natural structure for symbol separation and a classical approach to communication waveform design \cite{chang1966orthogonal}. Different orthogonal bases produce different waveform structures, and the choice of basis also affects the correlation response and ranging performance when these waveforms are used for sensing. The modulation basis is therefore a natural design variable for studying the sensing performance of this class of communication waveforms.

Matched filtering is a common receiver processing method for ranging, in which the received echoes are correlated with the known transmitted waveform to identify target delays {\cite{levanon2004radar}}. The waveform autocorrelation is central to this processing, since a delayed and scaled replica of the transmitted signal produces a correspondingly shifted and scaled autocorrelation response, with its main peak aligned to the target delay and its sidelobes extending to other candidate delays. When several targets are present, their correlation responses superpose, and the sidelobes of a strong target may obscure the main peak of a weaker target at a nearby delay \cite{liu2025ofdm}. Reducing these sidelobes lowers the response contributed by the strong target at candidate delays where weaker targets may be located.

We are interested in sensing targets whose echo delays extend beyond the duration of the \ac{CP} used for communication. A \ac{CP} that covers the echo delays supports a circular shift representation of the useful block after \ac{CP} removal. Extending the sensing range beyond this protected interval leads to linear matched filtering over a longer observation interval \cite{liu2025ofdm}. We study the response of an isolated finite transmitted block, represented by \(N\) samples and zero outside that interval, with the receiver collecting the samples needed for the candidate delays. Linear correlation of this block gives the aperiodic autocorrelation, which contains \(N-k\) overlapping sample pairs at a positive lag \(1\le k<N\). The target delay locates the response in the observation interval, while the correlation lag measures an offset from its peak. At zero lag, the autocorrelation equals the block energy. Since unitary modulation preserves that energy, comparisons between modulation matrices for the same information symbols keep the autocorrelation peak height fixed.

The sidelobes vary with the communication data, and their squared magnitudes quantify the energy in each realization. Averaging these squared magnitudes over the data retains the contributions of realizations with different phases and defines the \ac{ESL} at each lag \cite{liu2025ofdm}. For a target with unit reflection gain, the \ac{ESL} at lag \(k\) is the expected squared magnitude of that target's contribution to the unnormalized matched filter output at an offset of \(k\) samples from its true delay. It therefore measures average sidelobe leakage at a prescribed relative delay.

For a prescribed interval of relative delays, the preferred design objective is to find a single modulation that simultaneously minimizes the \ac{ESL} at every lag in that interval. Such a modulation would provide the smallest expected sidelobe leakage at each possible target separation within the interval, while also minimizing every nonnegative weighted sum of the corresponding \acp{ESL}. Simultaneous optimality at individual lags therefore gives the most direct control of the average sidelobes throughout the region of interest.

When no single modulation attains these individual minima simultaneously, the partial \ac{EISL} provides a relaxed objective that allows tradeoffs among the lags. For a selected nonempty subset of positive lags, the partial \ac{EISL} is the sum of the \acp{ESL} over that subset, with equal weight assigned to each selected lag. In this paper, we focus on the first \(r\) positive lags, namely the consecutive set \(\{1,\ldots,r\}\), where \(1\le r<N\), which we call a positive lag prefix. This objective is particularly relevant when the detection range is less than the total duration of each transmission block. By linearity of expectation, the partial \ac{EISL} over this prefix equals the average total sidelobe energy in the corresponding relative delay window across transmitted blocks. Minimizing it allows a reduction at one lag to compensate for an increase at another. The prefix length \(r\) specifies the relative delay range of interest, and \(r=N-1\) gives the full \ac{EISL}. Since the autocorrelation has equal magnitudes at opposite lags, the symmetric window \(1\le |k|\le r\) has twice the expected sidelobe energy of its positive half and the same minimizing modulations.

This hierarchy of objectives leads to two related questions for unitary modulation with a fixed symbol distribution. We first ask whether a common modulation can attain the minimum \ac{ESL} at every lag in a prescribed interval, and which individual lags admit a given modulation as a global minimizer. For partial \ac{EISL} minimization, we ask which modulation minimizes this sum over a prescribed prefix and whether the same choice remains optimal for every prefix length. Answering these questions establishes both the extent of simultaneous optimality at individual lags and the best achievable total sidelobe energy over a relative delay window.

\subsection{Related Work and Remaining Questions}

Deterministic sequence design controls correlation sidelobes by choosing the transmitted samples {\cite{golomb2005signal}}. Complementary sequence pairs achieve cancellation through the sum of their aperiodic autocorrelations \cite{golay1961complementary}, while suitable polyphase constructions yield zero periodic autocorrelation at every nonzero cyclic lag \cite{chu1972polyphase}. Iterative methods extend the design to prescribed lag intervals and additional sequence constraints \cite{stoica2009unimodular,soltanalian2012correlation}, with surrogate minimization providing a direct approach to integrated sidelobe optimization \cite{song2015autocorrelation}. Weighted criteria emphasize selected delays \cite{song2016weighted}, joint objectives balance peak and integrated sidelobes \cite{alaee2017pslisl}, and ambiguity function design treats regions containing both delay and Doppler shifts \cite{aubry2013ambiguity}. For communication signals, the design object becomes the modulation rule that generates a family of transmitted sequences from random information symbols.

The statistics of these random signals can be shaped through several design variables. Dedicated precoding optimizes sensing estimation losses averaged over the data \cite{lu2024precoding}, probabilistic constellation shaping adjusts the symbol distribution and its fourth moment \cite{du2024shaping}, and pulse shaping controls average ambiguity sidelobe energy under bandwidth and intersymbol interference constraints \cite{liao2025pulse}. Analyses of random orthogonal waveforms also characterize the normalized sidelobe energy summed over the complete discrete periodic delay and Doppler domain, which is invariant across modulation bases \cite{zhang2026ambiguity}. For a fixed symbol distribution, the energy assigned to individual correlation lags provides the basis for a more detailed comparison of orthogonal modulations.

Within this modulation class, \ac{OFDM} corresponds to the normalized \ac{IDFT} matrix \cite{weinstein1971dft}, and its use for radar ranging has been studied through the processing of communication signals \cite{sturm2009ofdm,sturm2011fusion}. Liu et al.\ \cite{liu2025ofdm} compared the expected correlation sidelobes over all orthonormal modulation bases. Under their symbol assumptions, \ac{OFDM} simultaneously attains the minimum \ac{ESL} at every nonzero periodic lag, thereby achieving the preferred objective throughout the periodic delay range and also globally minimizing the periodic \ac{EISL}. For the full aperiodic \ac{EISL}, the same work establishes local optimality and formulates the conjecture that \ac{OFDM} is globally optimal \cite[Theorem~3 and Conjecture~1]{liu2025ofdm}. Subsequent periodic analysis extends optimality at each lag to signals with Nyquist pulse shaping and oversampled delays \cite{liu2025iceberg}.

\begingroup
These results make \ac{OFDM} a natural candidate for the aperiodic optimization problem. {We first ask how to choose a unitary modulation that globally minimizes the \ac{ESL} at a prescribed aperiodic lag, and for which lags \ac{OFDM} attains that minimum. We also ask whether a common unitary modulation can attain the individual minima simultaneously over a prescribed prefix.} We then ask whether \ac{OFDM} globally minimizes the \ac{EISL} for every positive lag prefix and which matrices attain each minimum. Together, these questions determine when the preferred simultaneous objective is achievable and give the minimum partial \ac{EISL} for each prefix when tradeoffs among lags are allowed.
\par\endgroup

A fourth moment identity provides the starting point for this comparison \cite{liu2025ofdm}. It expresses each \ac{ESL} as a term independent of the modulation matrix plus a deterministic column correlation energy multiplied by \(\mu_4-2\), where \(\mu_4\) is the common fourth absolute moment of the symbols normalized to unit second absolute moment. Under the independent zero mean symbol model with vanishing unconjugated second moments, \(\mu_4<2\) describes the sub-Gaussian constellation class in the fourth moment classification of \cite[Definition~1]{liu2025ofdm}. This class includes uniform square \ac{QAM} and uniform \ac{PSK} with at least three points. The negative coefficient \(\mu_4-2\) therefore turns minimization of the expected sidelobes into maximization of the corresponding column correlation energy. Periodic shifts share the Fourier basis as an eigenbasis, giving a common diagonal representation for that comparison. Truncating shifts at the boundaries of a finite block introduces contributions that vary with the lag, and the aperiodic bounds require their joint control across the selected lags and the orthonormal modulation columns.

\subsection{Contributions of This Paper}

\begingroup
We consider independent complex information symbols with zero mean, vanishing unconjugated second moment, unit second absolute moment, and a common fourth absolute moment \(\mu_4<2\), corresponding to normalized sub-Gaussian constellations. This class includes the commonly used uniform square \ac{QAM} constellations and uniform \ac{PSK} constellations with at least three points \cite{liu2025ofdm}. For a fixed symbol distribution and a deterministic unitary modulation matrix, our contributions are as follows.

\begingroup
\begin{enumerate}
\item \textbf{Single lag optimality of \ac{OFDM}.} For every \(1\le k<N\), we prove that \ac{OFDM}, represented by the normalized \ac{IDFT}, globally minimizes the aperiodic \ac{ESL} at lag \(k\) among all unitary modulations if and only if \(k\mid N\).

\item \textbf{Global optimality of \ac{OFDM} for prefix partial \ac{EISL}.} For every \(1\le r<N\), we prove that \ac{OFDM} globally minimizes the partial \ac{EISL} over the positive lag prefix \(\{1,\ldots,r\}\) and give the exact minimum. The full prefix \(r=N-1\) resolves the global optimality conjecture in \cite[Conjecture~1]{liu2025ofdm}, while the shorter prefixes establish optimality for every initial relative delay window.

\item \textbf{Characterization and construction of optimal modulations.} For every nonempty prefix, we prove that a unitary matrix attaining the minimum partial \ac{EISL} must be obtained from the normalized \ac{IDFT} by arbitrary subcarrier phases, a subcarrier permutation, and a common linear phase across the subcarriers. This identifies the complete prefix minimizing family, which is independent of the prefix length. For an arbitrary single lag, we construct a globally optimal unitary matrix by applying a normalized \ac{IDFT} separately on each chain of sample indices spaced by that lag, and derive the exact minimum \ac{ESL} and its improvement over the full \ac{IDFT}.
\end{enumerate}

Together, these results give the simultaneous optimality criterion in Corollary~\ref{cor:tit-simultaneous-prefix}: a common unitary modulation attains all individual \ac{ESL} minima over \(\{1,\ldots,r\}\) if and only if every lag in that prefix divides \(N\). In particular, for \(N\ge3\), no unitary modulation simultaneously attains the individual minima at all positive lags.
\par\endgroup

The remainder of this paper is organized as follows. Section~\ref{sec:tit-model} presents the signal model, sidelobe criteria, fourth moment reduction, main optimality theorem, and the criterion for simultaneous minimization over a prefix. Section~\ref{sec:tit-proof-strategy} develops the common correlation energy gap formulation, its boundary decomposition, and the coefficient conditions used to control it. Section~\ref{sec:tit-lag-optimality} applies these conditions to single lags and prefixes and states the characterization of all prefix minimizers. Section~\ref{sec:tit-discussion} discusses {the feasibility of simultaneous optimality, the optimal modulation for a prescribed lag, the cost of improving that lag, and the design freedom within the prefix minimizing family.} Section~\ref{sec:tit-conclusion} concludes the paper. The proofs of the stated lemmas, propositions, theorem, and corollary are collected in Appendices~\ref{app:tit-fourth-moment}--\ref{app:tit-prefix-equality}, with one proof in each appendix.
\par\endgroup

\subsection*{Notation}

Bold lowercase and uppercase letters denote vectors and matrices, respectively, while the superscripts \(\mathrm T\) and \(\mathrm H\) denote transpose and Hermitian transpose. Complex conjugation is written as \(\overline z\) or \(z^*\), and \(\mathrm i=\sqrt{-1}\) denotes the imaginary unit. We use \(\|\cdot\|_2\) for the Euclidean norm of a vector, \(\|\cdot\|_F\) for the Frobenius norm of a matrix, and \(\operatorname{tr}(\cdot)\) and \(\operatorname{rank}(\cdot)\) for trace and rank. The matrix \(\mathbf I_N\) is the \(N\times N\) identity, \(\delta_{mn}\) is the Kronecker delta, and \(\mathbb E\) denotes expectation, with its random arguments specified by the model.

\section{Signal Model, Preliminaries, and Main Result}\label{sec:tit-model}

\subsection{Random Symbols and Unitary Modulation}\label{subsec:tit-symbol-model}

For a fixed integer \(N\ge2\), let \(\mathbf s=[s_0,\ldots,s_{N-1}]^{\mathrm T}\in\mathbb C^N\) contain independent communication symbols whose moments satisfy \begin{equation} \mathbb E s_i=\mathbb E s_i^2=0,\qquad \mathbb E|s_i|^2=1,\qquad \mathbb E|s_i|^4=\mu_4<2. \label{eq:pw-symbol-moments} \end{equation} The vanishing pseudo-variance, namely the unconjugated second moment, makes each zero mean symbol proper in the standard complex statistical sense \cite{neeser1993proper,picinbono1996second}, while the common second and fourth absolute moments ensure that the comparison between bases uses the same symbol statistics. Identical distributions are not required when the independence and common moments in \eqref{eq:pw-symbol-moments} hold. Following the fourth moment classification in \cite[Definition~1]{liu2025ofdm}, we call constellations satisfying \eqref{eq:pw-symbol-moments} sub-Gaussian constellations, whose fourth absolute moment is below the value \(2\) of a zero mean proper complex Gaussian symbol with unit power. This class includes commonly used normalized uniform square \ac{QAM} and uniform \(M\)-ary \ac{PSK} constellations with \(M\ge3\).

An admissible modulation is a deterministic unitary matrix \(\mathbf U=[\mathbf u_0,\ldots,\mathbf u_{N-1}]\in\mathbb C^{N\times N}\), which maps the symbols to the transmitted sample block according to \begin{equation} \mathbf y=\mathbf U\mathbf s,\qquad \mathbf U^{\mathrm H}\mathbf U=\mathbf U\mathbf U^{\mathrm H}=\mathbf I_N. \label{eq:tit-unitary-modulation} \end{equation} Thus \(\mathbf u_n\) specifies the contribution of symbol \(s_n\) to the \(N\) time samples. The matrix may be designed for the block length, the selected delay window, and the known symbol statistics, but it is fixed before the current symbol realization is generated. Unitarity and \eqref{eq:pw-symbol-moments} give {\begin{subequations} \begin{align} \|\mathbf y\|_2^2&=\|\mathbf s\|_2^2,\qquad \mathbb E\|\mathbf y\|_2^2=N, \label{eq:tit-block-energy}\\ \mathbb E[\mathbf y\mathbf y^{\mathrm H}]&=\mathbf U\mathbb E[\mathbf s\mathbf s^{\mathrm H}]\mathbf U^{\mathrm H}=\mathbf I_N. \label{eq:tit-sample-covariance} \end{align} \end{subequations}} Consequently, every admissible modulation preserves the energy of each realized block and has the same average energy and sample covariance, so differences in the sidelobe criteria arise from the higher order structure of the transmitted block.

The \ac{OFDM} choice is the normalized \ac{IDFT} matrix \(\mathbf V\), whose entries are \begin{equation} v_{m,n}=N^{-1/2}e^{2\pi\mathrm i mn/N},\qquad 0\le m,n<N. \label{eq:pw-idft-coefficients} \end{equation} The transform can be evaluated efficiently using \ac{FFT} algorithms \cite{weinstein1971dft}. In this model, \(\mathbf V\mathbf s\) already contains the transmitted time samples, and no additional Fourier synthesis is included in \eqref{eq:tit-unitary-modulation}.

\subsection{Aperiodic Correlation and Sidelobe Criteria}\label{subsec:tit-sensing-criterion}

We evaluate the correlation of the finite block by extending its samples as \(y_m=0\) for \(m\notin\{0,\ldots,N-1\}\). This convention fixes the behavior at the block boundaries, because a shift changes the number of overlapping sample pairs and does not wrap samples around the block.

To relate this correlation to ranging, consider a reflector with a fixed integer delay \(\tau\), a nonzero complex gain \(\beta\), and zero Doppler, for which the received samples have the form \begin{equation} z_m=\beta y_{m-\tau}+w_m. \label{eq:tit-echo-model} \end{equation} The noise is zero mean, independent of the communication symbols, and satisfies \(\mathbb E[w_m\overline{w_n}]=\sigma_w^2\delta_{mn}\). The sensing receiver observes the samples required for the selected delays and has access to the realized block \(\mathbf y\) \cite{xiong2023tradeoff,liu2025ofdm}, so that it can use that block as its matched filter reference \cite{turin1960matched}. At a candidate delay \(\tau+k\), with \(1\le k<N\), the matched filter output is \begin{equation} \begin{aligned} \sum_{m=0}^{N-1}\overline{y_m}z_{m+\tau+k}&=\beta\sum_{m=0}^{N-k-1}\overline{y_m}y_{m+k}\\ &\quad+\sum_{m=0}^{N-1}\overline{y_m}w_{m+\tau+k}. \end{aligned} \label{eq:tit-matched-response} \end{equation} In particular, the signal contribution is proportional to the aperiodic autocorrelation at the relative delay \(k\). At zero relative delay, this contribution is instead \(\beta\|\mathbf y\|_2^2=\beta\|\mathbf s\|_2^2\), whose dependence on the symbol realization is the same for every unitary modulation by \eqref{eq:tit-block-energy}.

For a prescribed positive lag, define its \ac{ESL} by \begin{equation} \operatorname{\acs*{ESL}}_k(\mathbf U):=\mathbb E\left|\sum_{m=0}^{N-k-1}\overline{y_m}y_{m+k}\right|^2,\qquad 1\le k<N, \label{eq:tit-single-esl} \end{equation} where the expectation is over the communication symbols. Our preferred objective over a selected interval is to find one modulation that minimizes \eqref{eq:tit-single-esl} simultaneously at every lag in that interval.

To allow tradeoffs among selected lags, define the partial \ac{EISL} over a nonempty subset \(\mathcal K\) of positive lags by \begin{equation} \mathcal E_{\mathcal K}(\mathbf U):=\sum_{k\in\mathcal K}\operatorname{\acs*{ESL}}_k(\mathbf U),\qquad \varnothing\ne\mathcal K\subseteq\{1,\ldots,N-1\}. \label{eq:tit-subset-eisl} \end{equation} Thus the partial \ac{EISL} sums the \acp{ESL} over exactly the lags in \(\mathcal K\). The zero lag is excluded because it corresponds to the mainlobe.

For the positive lag prefix \(\mathcal K=\{1,\ldots,r\}\), we use the shorter notation \begin{equation} \mathcal E_r(\mathbf U):=\mathcal E_{\{1,\ldots,r\}}(\mathbf U)=\sum_{k=1}^{r}\operatorname{\acs*{ESL}}_k(\mathbf U),\qquad 1\le r<N, \label{eq:tit-partial-eisl} \end{equation} If the sample spacing is \(T_s\), this prefix describes relative delays from \(T_s\) to \(rT_s\). The negative delay correlation is the complex conjugate of its positive counterpart, so a symmetric window containing both signs has expected sidelobe energy \(2\mathcal E_r(\mathbf U)\) and the same minimizing matrices. The endpoint \(r=N-1\) gives the full aperiodic \ac{EISL} \cite[Sec.~IV]{liu2025ofdm}.

For \ac{OFDM} to attain the preferred simultaneous optimum over the entire positive lag range, its matrix \(\mathbf V\) would have to satisfy, for every unitary \(\mathbf U\) and every \(1\le k<N\), \begin{equation} \operatorname{\acs*{ESL}}_k(\mathbf V)\le\operatorname{\acs*{ESL}}_k(\mathbf U). \label{eq:tit-pointwise-esl-question} \end{equation} If \eqref{eq:tit-pointwise-esl-question} held, multiplying its individual inequalities by arbitrary weights \(a_k\ge0\) and summing would give \begin{equation} \sum_{k=1}^{N-1}a_k\operatorname{\acs*{ESL}}_k(\mathbf V)\le\sum_{k=1}^{N-1}a_k\operatorname{\acs*{ESL}}_k(\mathbf U). \label{eq:tit-weighted-esl-implication} \end{equation} Thus simultaneous \ac{ESL} optimality over the full positive lag range would minimize every nonnegative weighted sum, including each prefix \ac{EISL} in \eqref{eq:tit-partial-eisl}. More generally, the same implication holds on any prescribed interval when the individual inequalities hold throughout that interval, with weights set to zero elsewhere. The single lag part of Theorem~\ref{cor:pw-eisl} shows that \ac{OFDM} satisfies \eqref{eq:tit-pointwise-esl-question} exactly at the lags dividing \(N\), so the full range requirement fails for \(N\ge3\). The prefix result establishes optimality for the partial \ac{EISL} objective by comparing the contributions of the lags jointly.

Under the echo model \eqref{eq:tit-echo-model}, independence and the noise covariance give the more explicit relation \begin{equation} \sum_{k=1}^{r}\mathbb E\left|\sum_{m=0}^{N-1}\overline{y_m}z_{m+\tau+k}\right|^2=|\beta|^2\mathcal E_r(\mathbf U)+rN\sigma_w^2. \label{eq:tit-noisy-sidelobe-sum} \end{equation} The expectation in \eqref{eq:tit-noisy-sidelobe-sum} includes both the symbols and the noise. Since the second term is independent of the modulation matrix, this model leads to the waveform design problems
\begin{subequations} \begin{align} \min_{\mathbf U^{\mathrm H}\mathbf U=\mathbf I_N}\operatorname{\acs*{ESL}}_k(\mathbf U), \label{eq:tit-esl-optimization}\\ \min_{\mathbf U^{\mathrm H}\mathbf U=\mathbf I_N}\mathcal E_r(\mathbf U), \label{eq:tit-eisl-optimization} \end{align} \end{subequations}
where the lag \(k\) or the prefix length \(r\), as well as the symbol distributions, is held fixed during the respective optimization. The single lag problem \eqref{eq:tit-esl-optimization} specifies the minimum sought at each lag in the preferred simultaneous design, while \eqref{eq:tit-eisl-optimization} minimizes the total expected sidelobe energy when tradeoffs among those lags are allowed.

\subsection{Fourth Moment Reduction}\label{subsec:tit-moment-preliminaries}

To analyze \eqref{eq:tit-esl-optimization} and \eqref{eq:tit-eisl-optimization}, we separate the symbol statistics from the dependence on the modulation basis. For \(\mathbf c\in\mathbb C^N\), \(1\le k<N\), and a nonempty lag set \(\mathcal K\subseteq\{1,\ldots,N-1\}\), define \begin{equation} \gamma_k(\mathbf c)=\sum_{i=0}^{N-k-1}\overline{c_i}c_{i+k},\qquad S_{\mathcal K}(\mathbf U)=\sum_{n=0}^{N-1}\sum_{k\in\mathcal K}|\gamma_k(\mathbf u_n)|^2. \label{eq:pw-correlation-objective} \end{equation} The first quantity gives the correlation appearing in \eqref{eq:tit-single-esl} when \(\mathbf c=\mathbf y\), whereas the second is a deterministic statistic of the modulation columns. The following fourth moment reduction also appears in \cite[Corollary~4]{liu2025ofdm} for identically distributed symbols, and the calculation in Appendix~\ref{app:tit-fourth-moment} shows that independence and the common moments in \eqref{eq:pw-symbol-moments} are sufficient for the same identity.

\begin{lemma}[Fourth moment reduction]\label{lem:tit-fourth-moment} For the independent symbols and unitary modulation specified in Section~\ref{subsec:tit-symbol-model}, {\begin{subequations} \begin{align} \operatorname{\acs*{ESL}}_k(\mathbf U)=\mathbb E|\gamma_k(\mathbf y)|^2&=(N-k)+(\mu_4-2)S_{\{k\}}(\mathbf U), \label{eq:pw-eisl-single-reduction}\\ \mathcal E_r(\mathbf U)&=\sum_{k=1}^{r}(N-k)+(\mu_4-2)S_{\{1,\ldots,r\}}(\mathbf U). \label{eq:pw-eisl-prefix-reduction} \end{align} \end{subequations}} \end{lemma} \begin{IEEEproof} {Please refer to Appendix~\ref{app:tit-fourth-moment}.} \end{IEEEproof}

The coefficient \(\mu_4-2\) in Lemma~\ref{lem:tit-fourth-moment} is the fourth joint cumulant of a normalized proper symbol. To make this connection explicit, the fourth cumulant formula \cite{mendel1991tutorial} and \eqref{eq:pw-symbol-moments} give \begin{equation} \begin{aligned} \operatorname{cum}(s_i,\overline{s_i},s_i,\overline{s_i})&=\mathbb E|s_i|^4-2\bigl(\mathbb E|s_i|^2\bigr)^2-\bigl|\mathbb E s_i^2\bigr|^2\\ &=\mu_4-2, \end{aligned} \label{eq:tit-fourth-cumulant} \end{equation} where \(\operatorname{cum}\) denotes the joint cumulant. Thus \eqref{eq:pw-eisl-single-reduction} retains the fourth order contribution of the symbol distribution explicitly, and \eqref{eq:tit-fourth-cumulant} explains why a common covariance does not by itself determine the expected sidelobes.

In light of \eqref{eq:pw-eisl-single-reduction} and \eqref{eq:pw-eisl-prefix-reduction}, the terms independent of \(\mathbf U\) are fixed by the block length and the selected delays, while all dependence on the basis is contained in \(S_{\{k\}}(\mathbf U)\) and \(S_{\{1,\ldots,r\}}(\mathbf U)\), respectively. The condition \(\mu_4<2\) in \eqref{eq:pw-symbol-moments} therefore gives {\begin{subequations} \begin{align} \operatorname*{arg\,min}_{\mathbf U^{\mathrm H}\mathbf U=\mathbf I_N}\operatorname{\acs*{ESL}}_k(\mathbf U)&=\operatorname*{arg\,max}_{\mathbf U^{\mathrm H}\mathbf U=\mathbf I_N}S_{\{k\}}(\mathbf U), \label{eq:tit-equivalent-single-optimization}\\ \operatorname*{arg\,min}_{\mathbf U^{\mathrm H}\mathbf U=\mathbf I_N}\mathcal E_r(\mathbf U)&=\operatorname*{arg\,max}_{\mathbf U^{\mathrm H}\mathbf U=\mathbf I_N}S_{\{1,\ldots,r\}}(\mathbf U), \label{eq:tit-equivalent-optimization} \end{align} \end{subequations}} thereby reducing the expected sidelobe problems to deterministic correlation inequalities over the unitary group. The same expansion remains valid when \(\mu_4<2\) is replaced by \(\mu_4<\infty\) and the other symbol assumptions are retained, so it also explains the role of the restriction: if \(\mu_4=2\), the expected criteria are independent of \(\mathbf U\), while for \(\mu_4>2\) the two maximizations become minimizations.

\subsection{Main Result and Interpretation}\label{subsec:tit-main-result}

\begingroup
The next theorem gives the exact minimum \ac{ESL} at every single lag, a unitary matrix attaining it, and the exact minimum partial \ac{EISL} over every prefix together with all prefix minimizers. Its proof in Appendix~\ref{app:tit-main-theorem} combines the deterministic bounds and their equality conditions.
\par\endgroup

\begin{theorem}[{Optimal unitary modulation for single lags and prefixes}]\label{cor:pw-eisl} For the signal model \eqref{eq:tit-unitary-modulation} with independent symbols satisfying \eqref{eq:pw-symbol-moments}, the following statements hold.
\begin{enumerate}
\item { For a prescribed lag \(1\le k<N\), write \(N=qk+h\), where \(q\ge1\) and \(0\le h<k\). The exact minimum is

\begin{equation}
 \min_{\mathbf U^{\mathrm H}\mathbf U=\mathbf I_N}\operatorname{\acs*{ESL}}_k(\mathbf U)
 =(N-k)+(\mu_4-2)\left(N-2k+\frac{h}{q+1}+\frac{k-h}{q}\right).
 \label{eq:tit-optimal-single-esl}
\end{equation}

It is attained by the unitary matrix \(\mathbf U_{\rm blk}\) constructed in \eqref{eq:pw-block-competitor-matrix}, which applies a normalized \ac{IDFT} separately on each chain of sample indices spaced by \(k\). Its improvement over the full \ac{IDFT} \(\mathbf V\) is

\begin{equation}
 \operatorname{\acs*{ESL}}_k(\mathbf V)
 -\min_{\mathbf U^{\mathrm H}\mathbf U=\mathbf I_N}\operatorname{\acs*{ESL}}_k(\mathbf U)
 =(2-\mu_4)\frac{h(k-h)}{Nq(q+1)}.
 \label{eq:tit-single-esl-gain}
\end{equation}

Consequently, the full \ac{IDFT} minimizes the \ac{ESL} at lag \(k\) if and only if \(k\mid N\).}
\item Every unitary \(\mathbf U\) and every \(1\le r<N\) satisfy \begin{equation} \mathcal E_r(\mathbf U)\ge\sum_{k=1}^{r}(N-k)+\frac{\mu_4-2}{N}\sum_{k=1}^{r}(N-k)^2=\mathcal E_r(\mathbf V). \label{eq:tit-optimal-partial-eisl} \end{equation} For any fixed \(1\le r<N\), equality holds in \eqref{eq:tit-optimal-partial-eisl} if and only if \begin{equation} u_{m,n}=\frac{e^{\mathrm i(\alpha_n+m\theta)}}{\sqrt N}e^{2\pi\mathrm i m\sigma(n)/N},\qquad 0\le m,n<N, \label{eq:tit-optimal-modulation-family} \end{equation} where \(\theta,\alpha_n\in\mathbb R\) and \(\sigma\) is a permutation of \(\{0,\ldots,N-1\}\).
\end{enumerate} \end{theorem}
\begin{IEEEproof} {Please refer to Appendix~\ref{app:tit-main-theorem}.} \end{IEEEproof}

\begingroup
For any matrix \(\mathbf U\) of the form \eqref{eq:tit-optimal-modulation-family}, direct substitution into \eqref{eq:pw-correlation-objective} gives
{\begin{subequations}
\begin{align}
 \gamma_k(\mathbf u_n)
 &=\frac{N-k}{N}e^{\mathrm i k(\theta+2\pi\sigma(n)/N)},
 \label{eq:tit-family-correlations}\\
 S_{\{k\}}(\mathbf U)&=\frac{(N-k)^2}{N}=S_{\{k\}}(\mathbf V),
 \qquad 1\le k<N.
 \label{eq:tit-family-single-energy}
\end{align}
\end{subequations}}
Thus the phase and permutation operations in \eqref{eq:tit-optimal-modulation-family} preserve every single lag correlation energy. Together with Theorem~\ref{cor:pw-eisl}, these identities determine when the preferred simultaneous objective is achievable by any unitary modulation.

\begin{corollary}[Simultaneous optimality over a prefix]\label{cor:tit-simultaneous-prefix}
Under the assumptions of Theorem~\ref{cor:pw-eisl}, fix \(1\le r<N\). There exists a unitary matrix \(\mathbf U_*\) satisfying
\begin{equation}
 \operatorname{\acs*{ESL}}_k(\mathbf U_*)
 =\min_{\mathbf U^{\mathrm H}\mathbf U=\mathbf I_N}
    \operatorname{\acs*{ESL}}_k(\mathbf U),\qquad 1\le k\le r,
 \label{eq:tit-simultaneous-prefix}
\end{equation}
if and only if
\begin{equation}
 k\mid N\qquad\text{for every }k\in\{1,\ldots,r\}.
 \label{eq:tit-simultaneous-divisibility}
\end{equation}
When \(N\ge3\), no unitary matrix simultaneously minimizes the \ac{ESL} at all positive lags.
\end{corollary}
\begin{IEEEproof} {Please refer to Appendix~\ref{app:tit-simultaneous-prefix}.} \end{IEEEproof}

Corollary~\ref{cor:tit-simultaneous-prefix} identifies the prefixes on which the individual minima can be attained by a common modulation. For all other prefixes, Theorem~\ref{cor:pw-eisl} gives the exact minimum partial \ac{EISL} \(\mathcal E_r\) and every matrix attaining it. The family in \eqref{eq:tit-optimal-modulation-family} is independent of \(r\), so selecting a different prefix preserves the full set of \ac{EISL} minimizers.
\par\endgroup

Taking \(r=N-1\) in \eqref{eq:tit-optimal-partial-eisl} gives \begin{equation} \min_{\mathbf U^{\mathrm H}\mathbf U=\mathbf I_N}\mathcal E_{N-1}(\mathbf U)=\frac{N(N-1)}2+(\mu_4-2)\frac{(N-1)(2N-1)}6. \label{eq:tit-full-eisl-value} \end{equation} This establishes the global aperiodic \ac{EISL} optimality conjectured in \cite[Conjecture~1]{liu2025ofdm}, and \eqref{eq:tit-optimal-partial-eisl} further shows that the same modulation remains optimal when only the first \(r\) positive delays are retained. For constant modulus symbols satisfying \eqref{eq:pw-symbol-moments}, \(\mu_4=1\), so the minimum has the particularly simple form \begin{equation} \mathcal E_r(\mathbf V)=\frac1N\sum_{k=1}^{r}k(N-k). \label{eq:tit-constant-modulus-eisl} \end{equation} The positive value in \eqref{eq:tit-constant-modulus-eisl} reflects the finite block boundaries even when every data symbol has the same magnitude.

\begingroup
\section{Correlation Energy Gaps and Coefficient Conditions}\label{sec:tit-proof-strategy}

\subsection{\acs*{IDFT} Benchmark and Correlation Energy Gap}\label{subsec:tit-idft-gap}

The fourth moment reductions \eqref{eq:pw-eisl-single-reduction} and \eqref{eq:pw-eisl-prefix-reduction} reduce Theorem~\ref{cor:pw-eisl} to upper bounds on the column correlation energies in \eqref{eq:pw-correlation-objective}. We first compute their values for the normalized \ac{IDFT} \(\mathbf V=[\mathbf v_0,\ldots,\mathbf v_{N-1}]\). Substituting \eqref{eq:pw-idft-coefficients} into \eqref{eq:pw-correlation-objective} gives
\begin{subequations}
\begin{align}
 \gamma_k(\mathbf v_n)&=\sum_{m=0}^{N-k-1}\frac1N e^{2\pi\mathrm i kn/N}
 =\frac{N-k}{N}e^{2\pi\mathrm i kn/N},
 \label{eq:tit-idft-correlation}\\
 S_{\mathcal K}(\mathbf V)&=N\sum_{k\in\mathcal K}\frac{(N-k)^2}{N^2}
 =\frac1N\sum_{k\in\mathcal K}(N-k)^2.
 \label{eq:tit-idft-benchmark}
\end{align}
\end{subequations}
For a nonempty lag set \(\mathcal K\subseteq\{1,\ldots,N-1\}\), define the correlation energy gap by
\begin{equation}
 B_{\mathcal K}(\mathbf U):=S_{\mathcal K}(\mathbf V)-S_{\mathcal K}(\mathbf U)
 =\frac1N\sum_{k\in\mathcal K}(N-k)^2-S_{\mathcal K}(\mathbf U).
 \label{eq:pw-deficit-definition}
\end{equation}
In particular, subtracting the fourth moment identities for \(\mathbf U\) and \(\mathbf V\) gives
{\begin{subequations}
\begin{align}
 \operatorname{\acs*{ESL}}_k(\mathbf U)-\operatorname{\acs*{ESL}}_k(\mathbf V)
 &=(2-\mu_4)B_{\{k\}}(\mathbf U),
 \label{eq:tit-single-gap-equivalence}\\
 \mathcal E_r(\mathbf U)-\mathcal E_r(\mathbf V)
 &=(2-\mu_4)B_{\{1,\ldots,r\}}(\mathbf U).
 \label{eq:tit-prefix-gap-equivalence}
\end{align}
\end{subequations}}
Since \(2-\mu_4>0\) by \eqref{eq:pw-symbol-moments}, the sign of each gap directly determines whether the candidate modulation improves on \ac{OFDM}. The proof of Theorem~\ref{cor:pw-eisl} therefore has three tasks:
\begin{enumerate}
\item {For every single lag \(k\), determine the exact maximum of \(S_{\{k\}}(\mathbf U)\), construct a maximizing unitary matrix, and compare its value with the full \ac{IDFT} benchmark \eqref{eq:tit-idft-benchmark}.}
\item For every \(1\le r<N\), establish \(B_{\{1,\ldots,r\}}(\mathbf U)\ge0\) for every unitary \(\mathbf U\).
\item For each such prefix, prove that a zero gap forces \(\mathbf U\) to have the form \eqref{eq:tit-optimal-modulation-family}, and verify that every matrix of this form attains a zero gap.
\end{enumerate}
Sections~\ref{subsec:pw-deficit}--\ref{subsec:tit-coefficient-solvability} develop the gap decomposition and the coefficient conditions shared by these tasks. Section~\ref{sec:tit-lag-optimality} then treats single lags and prefixes and states the equality characterization. Appendix~\ref{app:tit-main-theorem} combines the resulting bounds and equality conditions to prove Theorem~\ref{cor:pw-eisl}.
\par\endgroup

\subsection{{Gap Decomposition and Boundary Power Imbalances}}\label{subsec:pw-deficit}

\begingroup
Section~\ref{subsec:tit-idft-gap} reduces the required sidelobe bounds to nonnegativity of the correlation energy gap \(B_{\mathcal K}(\mathbf U)\) in \eqref{eq:pw-deficit-definition}. We now decompose this gap into a phase loss and a contribution from the magnitudes of the modulation entries. By \eqref{eq:pw-deficit-definition},
\par\endgroup
\begin{equation}
 N^2B_{\mathcal K}(\mathbf U)
 =\sum_{n=0}^{N-1}\sum_{k\in\mathcal K}
       \left[(N-k)^2-N^2|\gamma_k(\mathbf u_n)|^2\right].
 \label{eq:pw-deficit-entrywise}
\end{equation}
We first examine each summand in \eqref{eq:pw-deficit-entrywise}. Because these summands can be negative, we will use the sums over columns and lags to obtain the desired bound.

For every column \(\mathbf u_n\), denote its normalized magnitudes and powers by
\begin{subequations}
\begin{align}
 x_i^{(n)}&=\sqrt N|u_{i,n}|,\qquad
 p_i^{(n)}=(x_i^{(n)})^2=N|u_{i,n}|^2,\qquad 0\le i,n<N,
 \label{eq:pw-normalized-profile}\\
 \sum_{i=0}^{N-1}p_i^{(n)}
 &=N\sum_{i=0}^{N-1}|u_{i,n}|^2=N.
 \label{eq:pw-profile-normalization}
\end{align}
\end{subequations}
Equal coefficient magnitudes correspond to \(x_i^{(n)}=p_i^{(n)}=1\).
Since \(\mathbf x^{(n)}\) is real and nonnegative, the correlation
in \eqref{eq:pw-correlation-objective} satisfies
\begin{equation}
 \gamma_k(\mathbf x^{(n)})
 =\sum_{i=0}^{N-k-1}x_i^{(n)}x_{i+k}^{(n)}
 =N\sum_{i=0}^{N-k-1}|u_{i,n}||u_{i+k,n}|\ge0.
 \label{eq:pw-magnitude-correlation}
\end{equation}
The triangle inequality therefore gives
\begin{equation}
 N|\gamma_k(\mathbf u_n)|\le\gamma_k(\mathbf x^{(n)}),
 \qquad
 \gamma_k(\mathbf x^{(n)})^2-N^2|\gamma_k(\mathbf u_n)|^2\ge0.
 \label{eq:pw-phase-relaxation}
\end{equation}
Using \eqref{eq:pw-phase-relaxation}, we may decompose the summand of
\eqref{eq:pw-deficit-entrywise} as
\begin{equation}
 \begin{split}
 (N-k)^2-N^2|\gamma_k(\mathbf u_n)|^2
 ={}&\left[(N-k)^2-\gamma_k(\mathbf x^{(n)})^2\right]\\
 &+\left[\gamma_k(\mathbf x^{(n)})^2
                         -N^2|\gamma_k(\mathbf u_n)|^2\right].
 \end{split}
 \label{eq:pw-phase-magnitude-split}
\end{equation}
The second bracket in \eqref{eq:pw-phase-magnitude-split} is the
nonnegative loss caused by phase cancellation.  The first bracket
depends only on the magnitudes, and we will express it through
differences between entries separated by \(k\) lags.

Next, for a general nonnegative vector \(\mathbf x\) with
\(\sum_{i=0}^{N-1}x_i^2=N\), write \(p_i=x_i^2\) and define
{\begin{subequations}
\begin{align}
 E_k(\mathbf x)&=\sum_{i=0}^{N-k-1}(x_i-x_{i+k})^2,\qquad
 H_k(\mathbf x)=\sum_{i=0}^{N-k-1}(x_i+x_{i+k})^2,
 \label{eq:pw-EH-definition}\\
 A_k(\mathbf x)&=\sum_{i=0}^{N-k-1}(p_i+p_{i+k}),\qquad
 b_k(\mathbf x)=2(N-k)-A_k(\mathbf x).
 \label{eq:pw-A-definition}
\end{align}
\end{subequations}}
Their values on the columns defined in \eqref{eq:pw-normalized-profile}
are denoted by
\begin{equation}
 \begin{gathered}
 E_k^{(n)}=E_k(\mathbf x^{(n)}),\qquad
 H_k^{(n)}=H_k(\mathbf x^{(n)}),\\
 A_k^{(n)}=A_k(\mathbf x^{(n)}),\qquad
 b_k^{(n)}=b_k(\mathbf x^{(n)}).
 \end{gathered}
 \label{eq:pw-column-quantity-notation}
\end{equation}
For a fixed argument \(\mathbf x\), we also write \(E_k,H_k,A_k,b_k\), suppressing the superscripts. Here \(E_k\) measures the differences of magnitudes at lag \(k\), while \(b_k\) measures the deviation of \(A_k\) from the value \(2(N-k)\) corresponding to equal magnitudes. Its interpretation as a boundary imbalance is later derived in \eqref{eq:pw-boundary-form}. At this point, we have the following result.

\begin{lemma}[Identity for the magnitude contribution]\label{lem:pw-gap}
For every nonnegative \(\mathbf x\) with
\(\sum_{i=0}^{N-1}x_i^2=N\), and every \(1\le k<N\),
\begin{equation}
 (N-k)^2-\gamma_k(\mathbf x)^2
 =(N-k)b_k+\frac14\left(E_kH_k-b_k^2\right).
 \label{eq:pw-gap-identity}
\end{equation}
\end{lemma}
\begin{IEEEproof} {Please refer to Appendix~\ref{app:tit-magnitude-identity}.} \end{IEEEproof}

We now sum \eqref{eq:pw-gap-identity} over the columns using \eqref{eq:pw-column-quantity-notation}.  The row normalization of unitary matrices gives
\begin{equation}
 \sum_{n=0}^{N-1}p_i^{(n)}
 =N\sum_{n=0}^{N-1}|u_{i,n}|^2
 =N(\mathbf U\mathbf U^{\mathrm H})_{ii}=N.
 \label{eq:pw-row-accounting}
\end{equation}
Sum \eqref{eq:pw-A-definition} over the columns and apply \eqref{eq:pw-row-accounting} to obtain
\begin{subequations}
\begin{align}
 \sum_{n=0}^{N-1}A_k^{(n)}
 &=\sum_{i=0}^{N-k-1}
       \left(\sum_{n=0}^{N-1}p_i^{(n)}
             +\sum_{n=0}^{N-1}p_{i+k}^{(n)}\right)=2N(N-k),
 \label{eq:pw-A-column-sum}\\
 \sum_{n=0}^{N-1}b_k^{(n)}
 &=2N(N-k)-\sum_{n=0}^{N-1}A_k^{(n)}=0.
 \label{eq:pw-boundary-column-cancellation}
\end{align}
\end{subequations}
Thus the linear term \((N-k)b_k\) in \eqref{eq:pw-gap-identity}
cancels after summation over all columns.  Substituting
\eqref{eq:pw-phase-magnitude-split} and
\eqref{eq:pw-gap-identity} into \eqref{eq:pw-deficit-entrywise},
then using \eqref{eq:pw-boundary-column-cancellation}, yields
\begin{equation}
 \begin{split}
 N^2B_{\mathcal K}(\mathbf U)
 ={}&\frac14\sum_{n=0}^{N-1}\sum_{k\in\mathcal K}
       \left[E_k^{(n)}H_k^{(n)}-(b_k^{(n)})^2\right]\\
 &+\sum_{n=0}^{N-1}\sum_{k\in\mathcal K}
       \left[\gamma_k(\mathbf x^{(n)})^2
                         -N^2|\gamma_k(\mathbf u_n)|^2\right].
 \end{split}
 \label{eq:pw-deficit-decomposition}
\end{equation}
Because the second sum in \eqref{eq:pw-deficit-decomposition}
is nonnegative by \eqref{eq:pw-phase-relaxation}, a sufficient
condition for \(B_{\mathcal K}(\mathbf U)\ge0\) is
\begin{equation}
 \sum_{n=0}^{N-1}\sum_{k\in\mathcal K}(b_k^{(n)})^2
 \le\sum_{n=0}^{N-1}\sum_{k\in\mathcal K}E_k^{(n)}H_k^{(n)}.
 \label{eq:pw-deficit-sufficient-condition}
\end{equation}
This condition controls the chosen lags jointly and does not require \(b_k^2\le E_kH_k\) for every lag and every nonnegative $\mathbf{x}$.

Next, to further simplify \eqref{eq:pw-deficit-sufficient-condition}, we express \(b_k\) as a boundary imbalance and relate \(E_kH_k\) to differences of powers.  Since \(\sum_i p_i=N\), the definition \eqref{eq:pw-A-definition} gives
\begin{equation}
 A_k=\sum_{i=0}^{N-k-1}p_i+\sum_{i=k}^{N-1}p_i
 =2N-\sum_{i<k}p_i-\sum_{i\ge N-k}p_i.
 \label{eq:pw-A-boundaries}
\end{equation}
For any set \(\Set{K}\), let \(\mathbb{I}_\Set{K}\) denote its indicator,
which equals one on \(\Set{K}\) and zero elsewhere.  Define
\begin{equation}
 \eta_{\ell,i}
 =\mathbb{I}_{\{0,\ldots,\ell-1\}}(i)
  +\mathbb{I}_{\{N-\ell,\ldots,N-1\}}(i)
  -\frac{2\ell}{N},~~ 1\le\ell<N,~~0\le i<N.
 \label{eq:pw-boundary-coefficients}
\end{equation}
Equations~\eqref{eq:pw-A-definition} and \eqref{eq:pw-A-boundaries},
together with \(\sum_i p_i=N\), imply
\begin{equation}
 b_k=\sum_{i<k}p_i+\sum_{i\ge N-k}p_i-2k
     =\sum_{i=0}^{N-1}\eta_{k,i}p_i.
 \label{eq:pw-boundary-form}
\end{equation}
Observe that the two sums in \eqref{eq:pw-boundary-form} select the first \(k\)
and last \(k\) entries, respectively.  When they overlap, each shared entry
appears in both sums, as illustrated in
Figure~\ref{fig:pw-boundary-imbalance}.  If \(p_i=1\), their total
is \(2k\), so \(b_k\) measures the excess of the actual boundary
power over this reference value.  In particular,
\eqref{eq:pw-boundary-form} can also be written as
\begin{equation}
 b_k=\sum_{i=0}^{k-1}(p_i-1)
          +\sum_{i=N-k}^{N-1}(p_i-1).
 \label{eq:pw-boundary-excess}
\end{equation}

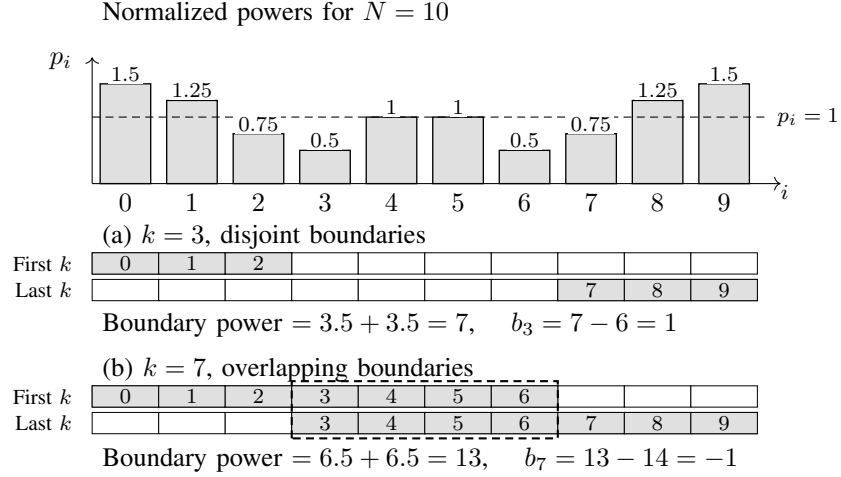
\begin{figure}[!t]
\centering
\begin{tikzpicture}[x=0.88cm,y=0.88cm,
  every node/.style={font=\small},
  every path/.style={}]
  \node[anchor=west] at (0,2.55) {Normalized powers for \(N=10\)};
  \draw[->] (0,0) -- (10.35,0);
  \draw[->] (0,0) -- (0,1.9);
  \node[anchor=east] at (-0.15,1.85) {\(p_i\)};
  \foreach \i/\p in {0/1.5,1/1.25,2/0.75,3/0.5,4/1,5/1,6/0.5,7/0.75,8/1.25,9/1.5} {
    \filldraw[fill=black!12,line width=0.5pt]
      (\i+0.12,0) rectangle (\i+0.88,\p);
    \node[below] at (\i+0.5,0) {\(\i\)};
  }
  \draw[densely dashed] (0,1) -- (10.15,1);
  \foreach \i/\p in {0/1.5,1/1.25,2/0.75,3/0.5,4/1,5/1,6/0.5,7/0.75,8/1.25,9/1.5} {
    \node[font=\scriptsize,above,fill=white,inner sep=1pt]
      at (\i+0.5,\p) {\(\p\)};
  }
  \node[anchor=west,font=\scriptsize] at (10.15,1) {\(p_i=1\)};
  \node[anchor=west,font=\scriptsize] at (10.2,-0.12) {\(i\)};

  \node[anchor=west] at (0,-0.8)
    {(a) \(k=3\), disjoint boundaries};
  \node[anchor=east,font=\scriptsize] at (-0.15,-1.2) {First \(k\)};
  \node[anchor=east,font=\scriptsize] at (-0.15,-1.6) {Last \(k\)};
  \foreach \i in {0,...,9} {
    \draw[fill=white,line width=0.3pt]
      (\i,-1.36) rectangle (\i+1,-1.04);
    \draw[fill=white,line width=0.3pt]
      (\i,-1.76) rectangle (\i+1,-1.44);
  }
  \foreach \i in {0,1,2} {
    \filldraw[fill=black!12,line width=0.3pt]
      (\i,-1.36) rectangle (\i+1,-1.04);
    \node[font=\scriptsize] at (\i+0.5,-1.2) {\(\i\)};
  }
  \foreach \i in {7,8,9} {
    \filldraw[fill=black!12,line width=0.3pt]
      (\i,-1.76) rectangle (\i+1,-1.44);
    \node[font=\scriptsize] at (\i+0.5,-1.6) {\(\i\)};
  }
  \node[anchor=west] at (0,-2.12)
    {Boundary power \(=3.5+3.5=7\), \(\quad b_3=7-6=1\)};

  \node[anchor=west] at (0,-2.8)
    {(b) \(k=7\), overlapping boundaries};
  \node[anchor=east,font=\scriptsize] at (-0.15,-3.2) {First \(k\)};
  \node[anchor=east,font=\scriptsize] at (-0.15,-3.6) {Last \(k\)};
  \foreach \i in {0,...,9} {
    \draw[fill=white,line width=0.3pt]
      (\i,-3.36) rectangle (\i+1,-3.04);
    \draw[fill=white,line width=0.3pt]
      (\i,-3.76) rectangle (\i+1,-3.44);
  }
  \foreach \i in {0,...,6} {
    \filldraw[fill=black!12,line width=0.3pt]
      (\i,-3.36) rectangle (\i+1,-3.04);
    \node[font=\scriptsize] at (\i+0.5,-3.2) {\(\i\)};
  }
  \foreach \i in {3,...,9} {
    \filldraw[fill=black!12,line width=0.3pt]
      (\i,-3.76) rectangle (\i+1,-3.44);
    \node[font=\scriptsize] at (\i+0.5,-3.6) {\(\i\)};
  }
  \draw[densely dashed,line width=0.8pt]
    (3,-3.83) rectangle (7,-2.97);
  \node[anchor=west] at (0,-4.15)
    {Boundary power \(=6.5+6.5=13\), \(\quad b_7=13-14=-1\)};
\end{tikzpicture}
\caption{Boundary imbalance in \eqref{eq:pw-boundary-form} and \eqref{eq:pw-boundary-excess}. The total power is \(N=10\), and the horizontal dashed
line marks the reference value \(p_i=1\). In (b), the overlap indices \(3,4,5,6\) are counted twice.}
\label{fig:pw-boundary-imbalance}
\end{figure}

Finally, define the total variation of the powers at lag \(k\) by
\begin{equation}
 V_k=\sum_{i=0}^{N-k-1}|p_i-p_{i+k}|.
 \label{eq:pw-V-definition}
\end{equation}
Since \(\mathbf x\) is nonnegative,
\begin{equation}
 |p_i-p_{i+k}|=|x_i-x_{i+k}|(x_i+x_{i+k}).
 \label{eq:pw-power-factorization}
\end{equation}
Cauchy--Schwarz applied to \eqref{eq:pw-V-definition} and
\eqref{eq:pw-power-factorization}, followed by
\eqref{eq:pw-EH-definition}, gives
\begin{equation}
 V_k^2
 \le\left(\sum_{i=0}^{N-k-1}(x_i-x_{i+k})^2\right)
     \left(\sum_{i=0}^{N-k-1}(x_i+x_{i+k})^2\right)
 =E_kH_k.
 \label{eq:pw-variation-bound}
\end{equation}
This reduces our remaining task to showing that
\begin{equation}\label{boundary_estimate}
\sum_{k\in\Set{K}}b_k^2 \leq \sum_{k\in\Set{K}}V_k^2
\end{equation}
holds for the considered sets $\Set{K}$, which will be detailed in the following sections.

\subsection{{A Sufficient Coefficient Condition}}\label{subsec:pw-contraction}

We now give a sufficient coefficient condition for the boundary estimate \eqref{boundary_estimate}. 

\begin{lemma}[Doubly substochastic contraction]\label{lem:pw-contraction}
Let \(\mathcal K\) be a nonempty set of allowed lags.  Suppose that there exist real coefficients $\{f_{i,k}^{(\ell)}\}$, independent of \(\mathbf p\), satisfying
\begin{equation}
 \sum_{i=0}^{N-1}\eta_{\ell,i}p_i = \sum_{k\in\mathcal K}\sum_{i=0}^{N-k-1}
 f^{(\ell)}_{i,k}(p_i-p_{i+k}),\qquad \ell\in\mathcal K,
 \label{eq:pw-expansion}
\end{equation}
for every real \(\mathbf p\).
For the normalized powers in
\eqref{eq:pw-profile-normalization}, the left side equals \(b_\ell\)
by \eqref{eq:pw-boundary-form}.
Assume there is a matrix \(\mathbf T=(t_{\ell k})\) and \(0\le\rho\le1\) such that
\begin{subequations}\label{eq:budgets_all}
\begin{align}
 |f^{(\ell)}_{i,k}|&\le t_{\ell k},
 \label{eq:pw-coefficient-capacity}\\
 \sum_{k\in\mathcal K}t_{\ell k}&\le\rho,\qquad
 \sum_{\ell\in\mathcal K}t_{\ell k}\le\rho.
 \label{eq:pw-coefficient-budgets}
\end{align}
\end{subequations}
In other words, the magnitudes
\(\lvert f_{i,k}^{(\ell)}\rvert\) are bounded by \(t_{\ell k}\)
uniformly in \(i\).  The matrix \(\mathbf T\) is nonnegative and
doubly substochastic, with every row sum and every column sum at
most \(\rho\). Then
{\begin{equation}
 \sum_{\ell\in\mathcal K}b_\ell^2
 \le\rho^2\sum_{k\in\mathcal K}V_k^2.
 \label{eq:pw-boundary-contraction}
\end{equation}}
\end{lemma}
\begin{IEEEproof} {Please refer to Appendix~\ref{app:tit-contraction}.} \end{IEEEproof}

\Needspace{5\baselineskip}
\subsection{{Solvability of the Coefficient Equations}}\label{subsec:tit-coefficient-solvability}

{We now determine when coefficients satisfying the expansion in Lemma~\ref{lem:pw-contraction} exist. Since \eqref{eq:pw-expansion} holds for every real \(\mathbf p\), comparing the coefficient of \(p_i\) on both sides gives the vertex balance}
\begin{equation}
 \sum_{\substack{k\in\mathcal K\\i+k<N}}f_{i,k}^{(\ell)}
 -\sum_{\substack{k\in\mathcal K\\k\le i}}f_{i-k,k}^{(\ell)}
 =\eta_{\ell,i},\qquad 0\le i<N,\quad \ell\in\mathcal K.
 \label{eq:pw-flow-vertex-balance}
\end{equation}
The condition \eqref{eq:pw-flow-vertex-balance} can be restated in the following compact matrix form
\begin{equation}\label{eq:balance_matrix}
\mathbf{A}_{\Set{K}} \mathbf{f}^{(\ell)}=\boldsymbol{\eta}_{\ell},~~\forall \ell \in\Set{K},
\end{equation}
where $\boldsymbol{\eta}_\ell=\begin{bmatrix}\eta_{\ell,0},&\eta_{\ell,1},&\dotsc,&\eta_{\ell,N-1}\end{bmatrix}^T\in\mathbb{R}^N$, with $\eta_{\ell,i}=\mathbb{I}_{\{0,\dotsc,\ell-1\}}(i)+\mathbb{I}_{\{N-\ell,\dotsc,N-1\}}(i)-\frac{2\ell}{N}$. The matrix $\mathbf{A}_{\Set{K}}$ is an incidence matrix, defined as
\begin{equation}\label{eq:pw-incidence-matrix}
[\mathbf{A}_{\Set{K}}]_{j,q(i,k)}=\begin{cases}
                                    1, & \mbox{$j=i$;}  \\
                                    -1, & \mbox{$j=i+k$;} \\
                                    0, & \mbox{otherwise},
                                  \end{cases}
\end{equation}
with the column index generator $q(i,k)$ defined as 
$$
q(i,k)=\sum_{h\in\Set{K},h<k}(N-h)+i.
$$
The vector $\mathbf{f}^{(\ell)}$ is given by $[\mathbf{f}^{(\ell)}]_{q(i,k)}=f_{i,k}^{(\ell)}$. Figure~\ref{fig:pw-incidence-N6-K24} illustrates the underlying undirected graph of \(\mathbf A_{\mathcal K}\) for \(N=6\) and
\(\mathcal K=\{2,4\}\).
\begin{figure}[!t]
\centering
\color{black}
\begin{tikzpicture}[x=1.8cm,y=0.75cm,
  font=\fontsize{12}{14}\selectfont,text=black,
  every path/.style={draw=black,line width=0.85pt}]
  \coordinate (v0) at (0,0);
  \coordinate (v1) at (1,0);
  \coordinate (v2) at (2,0);
  \coordinate (v3) at (3,0);
  \coordinate (v4) at (4,0);
  \coordinate (v5) at (5,0);

  \draw (v1) .. controls (1,2.75) and (5,2.75) .. (v5);
  \draw (v1) .. controls (1,1.15) and (3,1.15) .. (v3);
  \draw (v3) .. controls (3,1.15) and (5,1.15) .. (v5);
  \draw (v0) .. controls (0,-2.75) and (4,-2.75) .. (v4);
  \draw (v0) .. controls (0,-1.15) and (2,-1.15) .. (v2);
  \draw (v2) .. controls (2,-1.15) and (4,-1.15) .. (v4);

  \foreach \i in {0,1,2,3,4,5} {
    \filldraw[fill=white,draw=black] (v\i) circle (2.2pt);
    \node[right=5pt] at (v\i) {\(\i\)};
  }
  \pgfresetboundingbox
  \path[use as bounding box,draw=none]
    (-0.06,-2.12) rectangle (5.36,2.12);
\end{tikzpicture}
\caption{The underlying undirected graph for \(N=6\) and \(\mathcal K=\{2,4\}\). Each vertex is labelled by its index, and an edge joins \(i\) and \(j\) exactly when \(|i-j|\in\mathcal K\).}
\label{fig:pw-incidence-N6-K24}
\end{figure}
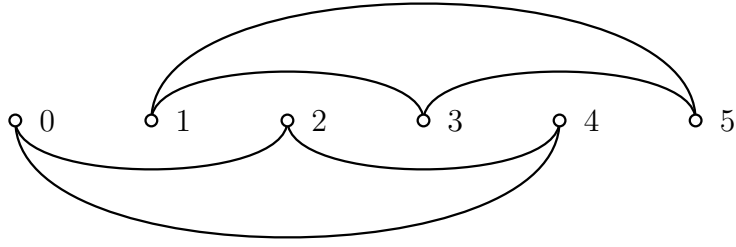

Now, the remaining task becomes to prove the existence of $\mathbf{f}^{(\ell)}$ satisfying \eqref{eq:balance_matrix} and having doubly substochastic entrywise upper bounds. We first apply the standard solvability criterion for incidence matrices to the boundary vectors in \eqref{eq:pw-boundary-coefficients}, yielding the following lemma.

\begin{lemma}[Solvability of the balance equations]\label{lem:pw-balance-solvability}
Let \(\mathcal K\subseteq\{1,\ldots,N-1\}\) be nonempty, and let
\(d=\gcd\{k:k\in\mathcal K\}\) be the greatest common divisor of its
elements.  The equations \eqref{eq:balance_matrix} have real
solutions for every \(\ell\in\mathcal K\) if and only if
\(d\mid N\) and \(\operatorname{rank}(\mathbf A_{\mathcal K})=N-d\).
\end{lemma}
\begin{IEEEproof} {Please refer to Appendix~\ref{app:tit-balance-solvability}.} \end{IEEEproof}

Next, we show that the conditions in Lemma~\ref{lem:pw-balance-solvability} are satisfied for various lag sets $\Set{K}$ of interest. 
\begin{enumerate}
\item Single lag $\Set{K}=\{k\}$, $k\mid N$: In this case, we have $d=k$, and hence $d\mid N$ is satisfied. The graph generated by $\mathbf{A}_{\Set{K}}$ consists of $k$ disjoint chains $\Set{C}_i=\{i,i+k,\dotsc,i+N-k\}$, $i=0,\dotsc,k-1$. Thus we have $s=k=d$ and hence ${\rm rank}(\mathbf{A}_{\Set{K}})=N-d$.
\item Any $\Set{K}\supseteq\{1\}$: In this case, we have $d=1$ since $1\in\Set{K}$. Thus $d\mid N$ holds. Furthermore, since every node $i$ is connected to $i+1$, the graph generated by $\mathbf{A}_{\Set{K}}$ is connected, thereby $s=d=1$, and ${\rm rank}(\mathbf{A}_{\Set{K}})=N-1$.
\end{enumerate}
\begingroup
\begingroup
The incidence equations are solvable for the dividing single lags and all prefixes identified above. Section~\ref{sec:tit-single} treats arbitrary single lags by estimating correlations on each coordinate chain, while Section~\ref{subsec:pw-prefix-construction} constructs coefficients satisfying the contraction condition for every prefix.
\par\endgroup
\par\endgroup

\section{{Optimality for Single Lags and Prefixes}}\label{sec:tit-lag-optimality}

\begingroup
We use the correlation estimates of Section~\ref{sec:tit-proof-strategy} to determine the optimum at every single lag and to bound every prefix. The single lag result includes an explicit minimizing modulation and its improvement over the full \ac{IDFT}, while the prefix construction leads to the complete characterization of prefix minimizers.
\par\endgroup

\subsection{{Single Lag Optimality}}\label{sec:tit-single}
\begingroup
Fix \(1\le k<N\), and write \(N=qk+h\) with \(q\ge1\) and \(0\le h<k\). The lag \(k\) partitions the sample indices into chains whose lengths differ by at most one:

\begin{subequations}
\begin{align}
 \mathcal C_j&=\{j+tk:t\ge0,\ j+tk<N\},\qquad 0\le j<k,
 \label{eq:pw-residue-chains}\\
 m_j:=|\mathcal C_j|&=\begin{cases}q+1,&0\le j<h,\\q,&h\le j<k.\end{cases}
 \label{eq:pw-chain-lengths}
\end{align}
\end{subequations}

Let \(\mathbf P\) reorder coordinates by listing \(\mathcal C_0,\ldots,\mathcal C_{k-1}\) in that order, with increasing indices within each chain. Explicitly, for every \(\mathbf z\in\mathbb C^N\),

\begin{equation}
 (\mathbf P\mathbf z)_{\sum_{a=0}^{j-1}m_a+t}=z_{j+tk},
 \qquad 0\le j<k,\quad 0\le t<m_j.
 \label{eq:tit-chain-permutation}
\end{equation}

For \(m\ge1\), denote the normalized \(m\)-point \ac{IDFT} by \(\mathbf F_m^{\mathrm H}\), with entries

\begin{equation}
 [\mathbf F_m^{\mathrm H}]_{t,n}=m^{-1/2}e^{2\pi\mathrm i tn/m},
 \qquad 0\le t,n<m.
 \label{eq:tit-chain-idft}
\end{equation}

Applying this transform separately on each chain gives the unitary modulation

\begin{equation}
 \mathbf U_{\rm blk}=\mathbf P^{\mathrm T}
 \operatorname{diag}(\mathbf F_{m_0}^{\mathrm H},\ldots,\mathbf F_{m_{k-1}}^{\mathrm H}).
 \label{eq:pw-block-competitor-matrix}
\end{equation}

The following proposition gives a matching upper bound for every unitary matrix, including matrices whose columns have nonzero entries on several chains.

\begin{proposition}[Exact single lag optimum]\label{prop:pw-single-lag} For the chains and matrix in \eqref{eq:pw-residue-chains}--\eqref{eq:pw-block-competitor-matrix},

\begin{align}
 \max_{\mathbf U^{\mathrm H}\mathbf U=\mathbf I_N}S_{\{k\}}(\mathbf U)
 &=S_{\{k\}}(\mathbf U_{\rm blk})
 =\sum_{j=0}^{k-1}\frac{(m_j-1)^2}{m_j}\nonumber\\
 &=N-2k+\frac{h}{q+1}+\frac{k-h}{q}.
 \label{eq:tit-single-correlation-optimum}
\end{align}

The difference from the full \ac{IDFT} benchmark is

\begin{equation}
 S_{\{k\}}(\mathbf U_{\rm blk})-S_{\{k\}}(\mathbf V)
 =\frac{h(k-h)}{Nq(q+1)}.
 \label{eq:tit-single-correlation-gain}
\end{equation}

Thus \(\mathbf V\) maximizes \(S_{\{k\}}\) if and only if \(h=0\), equivalently \(k\mid N\). \end{proposition}

\begin{IEEEproof} Please refer to Appendix~\ref{app:tit-single-lag}. \end{IEEEproof}
\par\endgroup

\subsection{{A Uniform Coefficient Bound for Prefixes}}\label{subsec:pw-prefix-construction}

Fix \(\mathcal K=\{1,\ldots,r\}\), where \(1\le r<N\). We first express the boundary imbalance using only allowed differences, which gives coefficients satisfying \eqref{eq:pw-expansion}. We then choose the paths used in this representation so that the coefficients also satisfy the budgets \eqref{eq:budgets_all}. Throughout this subsection, we fix a \(\ell\in\mathcal K\).

\subsubsection{{The Complete Graph and the Full Prefix}}\label{subsubsec:tit-full-prefix}
Let us define the reflection transformation as $R(i)=N-1-i$. We note that the boundary coefficients \eqref{eq:pw-boundary-coefficients} satisfy
\begin{subequations}\label{eq:pw-boundary}
\begin{align}
 \sum_{i=0}^{N-1}\eta_{\ell,i}&=0,
 \label{eq:pw-boundary-zero-sum}\\
 \eta_{\ell,R(i)}&=\eta_{\ell,i},\qquad
 |\eta_{\ell,i}-\eta_{\ell,j}|\le1.
 \label{eq:pw-boundary-symmetry}
\end{align}
\end{subequations}
To elaborate, reflection exchanges the two boundary indicators in \eqref{eq:pw-boundary-coefficients}, and their sum takes at most two consecutive integer values. This proves \eqref{eq:pw-boundary-symmetry}. For every real \(\mathbf p\), \eqref{eq:pw-boundary-zero-sum} gives the complete graph representation
\begin{equation}
 \begin{split}
 \sum_{i=0}^{N-1}\eta_{\ell,i}p_i
 &=\frac1N\sum_{0\le i<j<N}
       (\eta_{\ell,i}-\eta_{\ell,j})(p_i-p_j)\\
 &=\frac1N\sum_{k=1}^{N-1}\sum_{i=0}^{N-k-1}
       (\eta_{\ell,i}-\eta_{\ell,i+k})(p_i-p_{i+k}).
 \end{split}
 \label{eq:pw-pair-identity}
\end{equation}
For the full prefix \(r=N-1\), every edge in \eqref{eq:pw-pair-identity} is allowed. We can therefore take
\begin{equation}
 f_{i,k}^{(\ell)}
 =\frac{\eta_{\ell,i}-\eta_{\ell,i+k}}N,\qquad
 |f_{i,k}^{(\ell)}|\le t_{\ell k}:=\frac1N.
 \label{eq:pw-full-prefix-coefficients}
\end{equation}
The capacity estimate follows from \eqref{eq:pw-boundary-symmetry}, and the budgets are
\begin{equation}
 \sum_{k=1}^{N-1}t_{\ell k}
 =\sum_{\ell=1}^{N-1}t_{\ell k}
 =\frac{N-1}{N}<1.
 \label{eq:pw-full-prefix-budgets}
\end{equation}
Thus \eqref{eq:pw-expansion} and \eqref{eq:budgets_all} already hold for the full prefix case.

\subsubsection{{Replacing Forbidden Edges}}\label{subsubsec:tit-forbidden-edges}
Now suppose \(1\le r<N-1\).
The second sum in \eqref{eq:pw-pair-identity} includes terms with
\(k>r\). In what follows, we retain its terms with \(k\le r\), and transfer the contributions of the forbidden edges to allowed edges while preserving the identity.

First group each forbidden edge with its reflection.
Each group has the form \((a,R(b))\) and \((b,R(a))\), where
\begin{equation}
 0\le a\le b,\qquad
 R(b)-a>r\quad\Longleftrightarrow\quad a+b<N-r-1.
 \label{eq:pw-forbidden-pair}
\end{equation}
When \(a=b\), the edges coincide and their identical coefficient
\((\eta_{\ell,a}-\eta_{\ell,R(a)})/N\) is zero by
\eqref{eq:pw-boundary-symmetry}.
When \(a<b\), the two coefficients are opposite.
Since \(b<N-r-1<N-\ell\), the definition
\eqref{eq:pw-boundary-coefficients} and the reflection symmetry
\eqref{eq:pw-boundary-symmetry} give
\begin{equation}
 \frac{\eta_{\ell,a}-\eta_{\ell,R(b)}}N
 =-\frac{\eta_{\ell,b}-\eta_{\ell,R(a)}}N
 =\frac{\mathbb I_{\{a<\ell\le b\}}}{N}.
 \label{eq:pw-paired-demand}
\end{equation}
By \eqref{eq:pw-forbidden-pair} and \eqref{eq:pw-paired-demand},
the edges pairs with nonzero coefficients are indexed
by the vertex pairs in
\begin{equation}
 \Set{Q}_\ell
 =\{(a,b):0\le a<\ell\le b,\ a+b<N-r-1\}.
 \label{eq:pw-prefix-correction-pairs}
\end{equation}
Each \((a,b)\in\Set{Q}_\ell\) corresponds to the two original
edges \((a,R(b))\) and \((b,R(a))\).
Their coefficients in \eqref{eq:pw-pair-identity} are \(1/N\)
and \(-1/N\), respectively, by \eqref{eq:pw-paired-demand}.
We rewrite their combined contribution as
\begin{equation}
 \begin{split}
 &\frac1N\bigl[(p_a-p_{R(b)})-(p_b-p_{R(a)})\bigr]=\frac1N\bigl[(p_a-p_b)+(p_{R(a)}-p_{R(b)})\bigr].
 \end{split}
 \label{eq:pw-mirror-identity}
\end{equation}
Summing \eqref{eq:pw-mirror-identity} over \(\Set{Q}_\ell\), while retaining the allowed terms in \eqref{eq:pw-pair-identity}, we obtain
\begin{equation}
 \begin{split}
 N\sum_{i=0}^{N-1}\eta_{\ell,i}p_i
 ={}&\sum_{k=1}^{r}\sum_{i=0}^{N-k-1}
       (\eta_{\ell,i}-\eta_{\ell,i+k})(p_i-p_{i+k})\\
 &+\sum_{(a,b)\in\Set{Q}_\ell}
       \bigl[(p_a-p_b)+(p_{R(a)}-p_{R(b)})\bigr].
 \end{split}
 \label{eq:pw-prefix-folded-expansion}
\end{equation}
The first sum already uses only allowed edges.
The second sum gives the contributions of the prohibited edges which will later be expressed using allowed edges. 

For each \((a,b)\in\Set{Q}_\ell\), we may now choose an increasing path \(\Set{P}_{ab}\) with allowed steps. Its vertex sequence satisfies
\begin{equation}
 \begin{gathered}
 \Set{P}_{ab}=(v_0,\ldots,v_m),\qquad
 a=v_0<\cdots<v_m=b,\\
 1\le v_{j+1}-v_j\le r\quad(0\le j<m).
 \end{gathered}
 \label{eq:pw-correction-path-choice}
\end{equation}
Such paths exist because one can always consider paths with all edges having length $1$. Fix the paths independently of \(\mathbf p\). For every choice satisfying \eqref{eq:pw-correction-path-choice},
the two differences in \eqref{eq:pw-mirror-identity} can be decomposed as
\begin{subequations}\label{eq:pw-path-telescoping}
 \begin{align}
 p_a-p_b&=\sum_{j=0}^{m-1}(p_{v_j}-p_{v_{j+1}}),\label{eq:pw-path-telescoping1}\\
 p_{R(a)}-p_{R(b)}
 &=-\sum_{j=0}^{m-1}
       (p_{R(v_{j+1})}-p_{R(v_j)}).\label{eq:pw-path-telescoping2}
 \end{align}
\end{subequations}
Every edge on the right hand side of \eqref{eq:pw-path-telescoping} has length at most \(r\), including the reflected edges. Substitution into \eqref{eq:pw-prefix-folded-expansion} therefore removes all forbidden differences.

\subsubsection{{Updating Coefficients on Allowed Edges}}\label{subsubsec:tit-allowed-coefficients}
Although all edges with length exceeding $r$ can be replaced with allowed edges using the technique \eqref{eq:pw-path-telescoping}, many edges may be reused, and hence the coefficients may not meet the budget. In light of this, for the paths chosen in \eqref{eq:pw-correction-path-choice}, we define the reused time of each edge $(i,i+k)$ as
\begin{equation}
 c_{i,k}^{(\ell)}
 =\bigl|\{\Set{P}_{ab}:(a,b)\in\Set{Q}_\ell,\ 
       \Set{P}_{ab}\text{ contains }(i,i+k)\}\bigr|.
 \label{eq:pw-correction-count-definition}
\end{equation}
Now, note that the condition in \eqref{eq:pw-prefix-correction-pairs} implies
\begin{equation}
 0\le a<\ell\le b<N-r-1,\qquad (a,b)\in\Set{Q}_\ell.
 \label{eq:pw-correction-endpoints}
\end{equation}
Since the vertex indices in the paths are increasing, by \eqref{eq:pw-correction-endpoints}, $c_{i,k}^{(\ell)}$ vanish outside \(0\le i<i+k<N-r-1\). Summing \eqref{eq:pw-path-telescoping1} over the chosen paths gives
\begin{equation}
 \sum_{(a,b)\in\Set{Q}_\ell}(p_a-p_b)
 =\sum_{k=1}^{r}\sum_{i=0}^{N-k-1}
        c_{i,k}^{(\ell)}(p_i-p_{i+k}).
 \label{eq:pw-correction-expansion}
\end{equation}
Reflection sends the edge \((i,i+k)\) to \((N-1-i-k,N-1-i)\), and \eqref{eq:pw-path-telescoping2} gives
\begin{equation}
 \sum_{(a,b)\in\Set{Q}_\ell}(p_{R(a)}-p_{R(b)})
 =-\sum_{k=1}^{r}\sum_{i=0}^{N-k-1}
       c_{N-1-i-k,k}^{(\ell)}(p_i-p_{i+k}).
 \label{eq:pw-reflected-correction-expansion}
\end{equation}
If \(\ell\ge N-r-1\), then \(\Set{Q}_\ell\) is empty by \eqref{eq:pw-prefix-correction-pairs}, and all the reused times are zero, thus \eqref{eq:pw-correction-expansion} and \eqref{eq:pw-reflected-correction-expansion} still hold. Substituting \eqref{eq:pw-correction-expansion} and
\eqref{eq:pw-reflected-correction-expansion} into
\eqref{eq:pw-prefix-folded-expansion} proves
\eqref{eq:pw-expansion} with
\begin{equation}
 f_{i,k}^{(\ell)}
 =\frac{\eta_{\ell,i}-\eta_{\ell,i+k}
          +c_{i,k}^{(\ell)}-c_{N-1-i-k,k}^{(\ell)}}{N}.
 \label{eq:pw-unified-prefix-coefficients}
\end{equation}
The difference $\eta_{\ell,i}-\eta_{\ell,i+k}$ is the original contribution on the allowed edge. The term $c_{i,k}^{(\ell)}$ adds the contributions of the chosen paths, and the last term $c_{N-1-i-k,k}^{(\ell)}$ subtracts those of their reflections. Thus \eqref{eq:pw-unified-prefix-coefficients} satisfies the expansion identity for every path choice in \eqref{eq:pw-correction-path-choice}. 

\subsubsection{{Constructing Path Arrangements}}\label{subsubsec:tit-path-arrangements}
It now remains to choose paths for which the budgets \eqref{eq:budgets_all} also hold. The nonnegativity of $c_{i,k}^{(\ell)}$ and the bound \eqref{eq:pw-boundary-symmetry} imply
\begin{equation}
 \begin{split}
 |f_{i,k}^{(\ell)}|
 &\le\frac{1+\bigl|c_{i,k}^{(\ell)}
                        -c_{N-1-i-k,k}^{(\ell)}\bigr|}{N}\le\frac{1+\max\{c_{i,k}^{(\ell)},
                         c_{N-1-i-k,k}^{(\ell)}\}}{N}.
 \end{split}
 \label{eq:pw-path-count-reduction}
\end{equation}
Consequently, bounds on \(c_{i,k}^{(\ell)}\) give bounds on the
coefficients in \eqref{eq:pw-unified-prefix-coefficients}.
The following lemma chooses paths directly for the endpoint set
\(\Set{Q}_\ell\) and bounds these counts.

\begin{lemma}[Path arrangement with bounded edge congestion]\label{lem:pw-interval}
Fix integers \(N,r,\ell\) satisfying \(1\le\ell\le r<N\), and let \(\Set{Q}_\ell\) be the set in \eqref{eq:pw-prefix-correction-pairs}. There exists a path arrangement \(\{\Set{P}_{ab}:(a,b)\in\Set{Q}_\ell\}\), in which each path joins \(a\) to \(b\) with increasing vertex indices and edge lengths in \(\{1,\ldots,r\}\), such that the counts \(c_{i,k}^{(\ell)}\) defined in \eqref{eq:pw-correction-count-definition} satisfy
\begin{equation}
 0\le c_{i,k}^{(\ell)}\le
 \begin{cases}
 \max\{1,N-2r-1\},&k=\ell,\\
 1,&k\ne\ell,
 \end{cases}
 \label{eq:pw-correction-count-bound}
\end{equation}
for every permissible edge \((i,i+k)\) with \(1\le k\le r\) and \(i+k<N\). We have \(c_{i,k}^{(\ell)}=0\) whenever \(\ell\ge N-r-1\) or \(i+k\ge N-r-1\). Thus each edge of length different from \(\ell\) is used by at most one chosen path, while each edge of length \(\ell\) is used by at most \(\max\{1,N-2r-1\}\) chosen paths.
\end{lemma}
\begin{IEEEproof} {Please refer to Appendix~\ref{app:tit-path-arrangement}.} \end{IEEEproof}

{\color{black}
Figure~\ref{fig:pw-lemma4-path-arrangement} shows all paths for a concrete \(\Set{Q}_\ell\) and their edge reuse times.
}
\begin{figure}[!ht]
\centering
\color{black}
\begin{tikzpicture}[x=1cm,y=1cm,
  font=\fontsize{12}{14}\selectfont,text=black,
  every path/.style={draw=black,line width=0.85pt}]
  \node[anchor=west] at (-0.12,1.90)
    {(a) All paths for \(N=8,\ r=\ell=2\)};
  \node[anchor=west] at (-0.12,1.40)
    {\(\mathcal Q_2=\{(0,2),(0,3),(0,4),(1,2),(1,3)\}\)};

  \foreach \column/\row/\a/\b in
    {0/0/0/2,1/0/0/3,2/0/0/4,0/1/1/2,1/1/1/3} {
    \begin{scope}[shift={(\column*4.15,-\row*1.75)}]
      \node at (1.26,0.91) {\(\mathcal P_{\a\b}\)};
      \ifnum\a=0
        \draw (0,0) .. controls (0,0.70) and (1.26,0.70) .. (1.26,0);
        \ifnum\b=3
          \draw (1.26,0) -- (1.89,0);
        \fi
        \ifnum\b=4
          \draw (1.26,0) .. controls (1.26,0.70) and (2.52,0.70) .. (2.52,0);
        \fi
      \else
        \ifnum\b=2
          \draw (0.63,0) -- (1.26,0);
        \else
          \draw (0.63,0) .. controls (0.63,0.70) and (1.89,0.70) .. (1.89,0);
        \fi
      \fi
      \foreach \i in {0,...,4} {
        \filldraw[fill=white,draw=black] (\i*0.63,0) circle (2.2pt);
        \node[below=4pt] at (\i*0.63,0) {\(\i\)};
      }
    \end{scope}
  }
  \node[align=center] at (9.56,-1.40)
    {\((1,4)\notin\mathcal Q_2\)\\[3pt]
     \(1+4=5\not<5\)};

  \node[anchor=west] at (-0.12,-2.63)
    {(b) Edge reuse times \(c_{i,k}^{(2)}\) for all five paths};
  \begin{scope}[shift={(1.40,-3.98)},font=\fontsize{12}{14}\selectfont]
    \foreach \i in {0,...,4} {
      \coordinate (v\i) at (\i*2,0);
    }
    \draw (v0) .. controls (0,-1.18) and (4,-1.18) .. (v2)
      node[pos=0.5,fill=white,inner sep=2pt] {\(\mathbf{3}\)};
    \draw (v2) .. controls (4,-1.18) and (8,-1.18) .. (v4)
      node[pos=0.5,fill=white,inner sep=2pt] {\(\mathbf{1}\)};
    \draw (v1) .. controls (2,1.18) and (6,1.18) .. (v3)
      node[pos=0.5,fill=white,inner sep=2pt] {\(\mathbf{1}\)};
    \foreach \i/\j in {1/2,2/3} {
      \draw (v\i) -- (v\j)
        node[midway,fill=white,inner sep=2pt] {\(\mathbf{1}\)};
    }
    \foreach \i in {0,...,4} {
      \filldraw[fill=white,draw=black] (v\i) circle (2.2pt);
    }
    \foreach \i in {0,2,4} {
      \node[above=5pt] at (v\i) {\(\i\)};
    }
    \foreach \i in {1,3} {
      \node[below=5pt] at (v\i) {\(\i\)};
    }
  \end{scope}
  \node at (5.40,-5.43)
    {\(c_{i,1}^{(2)}\le1
       \qquad c_{i,2}^{(2)}\le\max\{1,8-2\cdot2-1\}=3\)};
\end{tikzpicture}
\caption{The complete path arrangement for
\(\Set{Q}_2\) when \(N=8\) and \(r=\ell=2\).
Panel (a) shows one path for every endpoint pair, read from left to right.
Panel (b) labels each used edge \((i,i+k)\) by its count
\(c_{i,k}^{(2)}\). The edge \((0,2)\) is used three times,
attaining the bound in \eqref{eq:pw-correction-count-bound},
and every other drawn edge is used once.}
\label{fig:pw-lemma4-path-arrangement}
\end{figure}
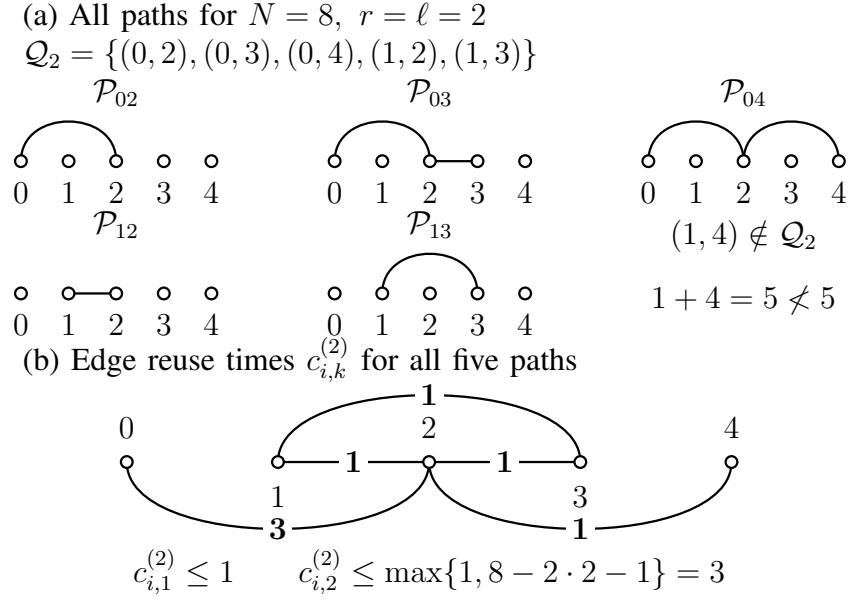

Now, we choose the paths supplied by Lemma~\ref{lem:pw-interval}.
The bound \eqref{eq:pw-correction-count-bound}, together with
the vanishing conditions in the lemma, applies to both counts
in \eqref{eq:pw-path-count-reduction}.
Using \(\max\{1,N-2r-1\}=1+(N-2r-2)_+\), we obtain
\begin{equation}
 |f_{i,k}^{(\ell)}|\le t_{\ell k}:=
 \frac1N\left[
 1+\mathbb I_{\{\ell<N-r-1\}}\mathbb I_{\{k<N-r-1\}}
   +(N-2r-2)_+\mathbb I_{\{\ell=k\}}
 \right].
 \label{eq:pw-unified-prefix-table}
\end{equation}
The matrix \(\mathbf T=(t_{\ell k})\) given by \eqref{eq:pw-unified-prefix-table} is nonnegative and symmetric. For every row, summing the entries in \eqref{eq:pw-unified-prefix-table} gives
\begin{equation}
 \begin{split}
 \sum_{k=1}^{r}t_{\ell k}
 &\le\frac{r+\min\{r,(N-r-2)_+\}+(N-2r-2)_+}{N}\\
 &=\frac{r+(N-r-2)_+}{N}
 =\frac{\max\{r,N-2\}}{N}<1.
 \end{split}
 \label{eq:pw-unified-prefix-budgets}
\end{equation}
Symmetry gives the same bound for every column. For \(r=N-1\), the same bound is \eqref{eq:pw-full-prefix-budgets}. Thus the chosen paths give coefficients \eqref{eq:pw-unified-prefix-coefficients} satisfying both the expansion \eqref{eq:pw-expansion} and the budgets \eqref{eq:budgets_all}.

By Lemma~\ref{lem:pw-contraction} and the reduction to
\eqref{boundary_estimate}, the gap
\eqref{eq:pw-deficit-definition} is therefore nonnegative,
which gives the prefix bound
\begin{equation}
 S_{\{1,\ldots,r\}}(\mathbf U)
 \le\frac1N\sum_{k=1}^{r}(N-k)^2,
 \qquad 1\le r<N,
 \label{eq:pw-prefix-bound}
\end{equation}
for every unitary \(\mathbf U\).

\subsection{{All Prefix Minimizers}}\label{subsec:pw-prefix-equality}

\begingroup
Equation~\eqref{eq:tit-idft-benchmark} shows that the normalized \ac{IDFT} attains the prefix bound \eqref{eq:pw-prefix-bound}. We now determine every unitary matrix attaining that bound.
\begin{proposition}[Necessity for prefix equality]\label{prop:unique}
Let \(1\le r<N\), and let \(\mathbf U\) be unitary. If equality holds in \eqref{eq:pw-prefix-bound}, then \(\mathbf U\) has the form \eqref{eq:tit-optimal-modulation-family}.
\end{proposition}
\par\endgroup

\begin{IEEEproof} {Please refer to Appendix~\ref{app:tit-prefix-equality}.} \end{IEEEproof}

\begingroup
\section{Discussion}\label{sec:tit-discussion}

Theorem~\ref{cor:pw-eisl} and Corollary~\ref{cor:tit-simultaneous-prefix} connect the selected delay window to the attainable sidelobe performance and the available modulation choices. In this section, we discuss some consequences of these results.

\subsection{When Simultaneous Optimality of All \texorpdfstring{\acp{ESL}}{ESLs} Is Achievable in a Prefix}\label{subsec:tit-discussion-simultaneous}

Corollary~\ref{cor:tit-simultaneous-prefix} determines whether a common modulation can attain the individual \ac{ESL} minima throughout a prefix. For example, a block length of \(N=12\) satisfies
\begin{equation}
 1\mid12,\qquad 2\mid12,\qquad 3\mid12,\qquad 4\mid12,\qquad 5\nmid12.
 \label{eq:tit-discussion-divisors}
\end{equation}
Consequently, \ac{OFDM} simultaneously minimizes the \ac{ESL} at lags \(1,2,3,4\). As in \eqref{eq:tit-weighted-esl-implication}, multiplying the individual inequalities by nonnegative weights and summing gives, for every unitary \(\mathbf U\),
\begin{equation}
 \sum_{k=1}^{4}a_k\operatorname{\acs*{ESL}}_k(\mathbf V)
 \le\sum_{k=1}^{4}a_k\operatorname{\acs*{ESL}}_k(\mathbf U),
 \qquad a_1,a_2,a_3,a_4\ge0.
 \label{eq:tit-discussion-weighted-example}
\end{equation}
Extending the prefix to lag five changes the feasibility of the preferred objective. By \eqref{eq:tit-discussion-divisors} and Corollary~\ref{cor:tit-simultaneous-prefix}, no unitary modulation attains all five individual minima simultaneously. Theorem~\ref{cor:pw-eisl} nevertheless gives \ac{OFDM} as a global minimizer of their sum, thereby identifying the best total expected sidelobe energy for this larger window.

\begingroup
\subsection{Optimal Modulation for a Prescribed Lag}\label{subsec:tit-discussion-single-optimum}

Proposition~\ref{prop:pw-single-lag} specifies a globally optimal modulation for every individual lag. In the construction \eqref{eq:pw-block-competitor-matrix}, samples separated by \(k\) belong to the same chain, and each chain uses an \ac{IDFT} of its own length. The chain lengths determine the exact optimum through \eqref{eq:tit-single-correlation-optimum}. For \(k\mid N\), all chains have the same length and both this construction and the full \ac{IDFT} attain the minimum \ac{ESL}. For \(k\nmid N\), the chains have two different lengths, and \eqref{eq:tit-single-esl-gain} gives the strict reduction achieved by the construction.

For the block length \(N=12\) used in Section~\ref{subsec:tit-discussion-simultaneous}, take \(k=5\). Then \(q=2\), \(h=2\), and the five chain lengths are \(3,3,2,2,2\). An optimal matrix is

\begin{equation}
 \mathbf U_{\rm blk}=\mathbf P^{\mathrm T}
 \operatorname{diag}(\mathbf F_3^{\mathrm H},\mathbf F_3^{\mathrm H},
 \mathbf F_2^{\mathrm H},\mathbf F_2^{\mathrm H},\mathbf F_2^{\mathrm H}),
 \label{eq:tit-discussion-single-example}
\end{equation}

with \(\mathbf P\) defined by \eqref{eq:tit-chain-permutation}. Equations~\eqref{eq:tit-single-correlation-optimum} and \eqref{eq:tit-single-esl-gain} give

\begin{subequations}
\begin{align}
 S_{\{5\}}(\mathbf U_{\rm blk})&=2\frac{(3-1)^2}{3}+3\frac{(2-1)^2}{2}
 =\frac{25}{6},
 \label{eq:tit-discussion-single-energy}\\
 \operatorname{\acs*{ESL}}_5(\mathbf V)-\operatorname{\acs*{ESL}}_5(\mathbf U_{\rm blk})
 &=\frac{2-\mu_4}{12}.
 \label{eq:tit-discussion-single-gain}
\end{align}
\end{subequations}

This modulation attains the lowest possible \ac{ESL} at lag five. The prefix bounds determine the accompanying changes at earlier lags, as discussed next.
\par\endgroup

\subsection{The Cost of Improving an Individual Lag}\label{subsec:tit-discussion-tradeoff}

For a fixed unitary matrix \(\mathbf U\), define its change in \ac{ESL} relative to the \ac{IDFT} by
\begin{equation}
 \Delta_k:=\operatorname{\acs*{ESL}}_k(\mathbf U)
             -\operatorname{\acs*{ESL}}_k(\mathbf V),\qquad 1\le k<N.
 \label{eq:tit-discussion-esl-change}
\end{equation}
The definition \eqref{eq:tit-partial-eisl} and the prefix bound \eqref{eq:tit-optimal-partial-eisl} imply
\begin{equation}
 \sum_{j=1}^{r}\Delta_j=\mathcal E_r(\mathbf U)-\mathcal E_r(\mathbf V)\ge0,
 \qquad 1\le r<N.
 \label{eq:tit-discussion-cumulative-change}
\end{equation}
If \(\mathbf U\) improves the \ac{ESL} at a lag \(k\), so that \(\Delta_k<0\), then \eqref{eq:tit-discussion-cumulative-change} with \(r=k\) yields
\begin{align}
 \sum_{j=1}^{k-1}\Delta_j
 &=\sum_{j=1}^{k}\Delta_j-\Delta_k\nonumber\\
 &\ge-\Delta_k>0.
 \label{eq:tit-discussion-earlier-lag-cost}
\end{align}
\begingroup
Thus every improvement at an individual lag carries a net increase in the \ac{ESL} over the preceding lags, whose total is at least the size of that improvement. For the optimal matrix in \eqref{eq:tit-discussion-single-example}, the reduction at lag five is \((2-\mu_4)/12\) by \eqref{eq:tit-discussion-single-gain}, so \eqref{eq:tit-discussion-earlier-lag-cost} requires a total increase of at least \((2-\mu_4)/12\) over lags one through four. This locates the cost of the single lag improvement and explains how the construction coexists with optimality of every prefix sum.
\par\endgroup

\subsection{{Design Freedom within the Prefix Minimizing Family}}\label{subsec:tit-discussion-design-freedom}

{The complete prefix minimizing \ac{OFDM} family in \eqref{eq:tit-optimal-modulation-family}} leaves the subcarrier phases, subcarrier permutation, and a common linear phase adjustable. A subcarrier phase changes the phase assigned to each subcarrier, a subcarrier permutation changes the mapping of symbols, and the common linear phase multiplies sample \(m\) by \(e^{\mathrm i m\theta}\). For every matrix in this family, \eqref{eq:tit-family-single-energy}, \eqref{eq:pw-eisl-single-reduction}, and \eqref{eq:tit-optimal-partial-eisl} give
\begin{subequations}
\begin{align}
 \operatorname{\acs*{ESL}}_k(\mathbf U)&=\operatorname{\acs*{ESL}}_k(\mathbf V),
 \qquad 1\le k<N,
 \label{eq:tit-discussion-family-esl}\\
 \mathcal E_r(\mathbf U)&=\mathcal E_r(\mathbf V)
 =\min_{\widetilde{\mathbf U}^{\mathrm H}\widetilde{\mathbf U}=\mathbf I_N}
       \mathcal E_r(\widetilde{\mathbf U}),\qquad 1\le r<N.
 \label{eq:tit-discussion-family-eisl}
\end{align}
\end{subequations}
The minimizing family therefore has a common \ac{ESL} profile, and changing the prefix length preserves the full set of optimal modulations. These adjustable parameters allow additional design requirements to be studied within the family while preserving the common \ac{ESL} profile in \eqref{eq:tit-discussion-family-esl} and every prefix \ac{EISL} minimum in \eqref{eq:tit-discussion-family-eisl}. For example, one may investigate how the quantiles of the sidelobes, instead of their expectations, can be further controlled via certain permutations of the subcarriers.

\section{Conclusion}\label{sec:tit-conclusion}
\begingroup
For independent proper symbols drawn from sub-Gaussian constellations under the stated moment assumptions, we have determined the exact minimum \ac{ESL} at every single lag and an attaining unitary modulation based on the coordinate chains. Its improvement over the full \ac{IDFT} vanishes exactly when the lag divides the block length. Simultaneous minimization over \(\{1,\ldots,r\}\) is possible exactly when every lag in that prefix divides \(N\). For partial \ac{EISL}, \ac{OFDM} attains the global minimum for every prefix, and all prefix minimizers form the same family independently of the prefix length. The correlation energy gap decomposition supplies both the estimate on each chain and the boundary control needed for the prefix results.
\par\endgroup
\par\endgroup

\appendices
\begingroup
\section{Proof of Lemma~\ref{lem:tit-fourth-moment}}
\label{app:tit-fourth-moment}
\begin{IEEEproof}

Independence,
zero means, and properness in \eqref{eq:pw-symbol-moments} imply
\begin{equation}
 \mathbb E[\overline{s_i}s_js_p\overline{s_q}]
 =\delta_{ij}\delta_{pq}+\delta_{ip}\delta_{jq}
             +(\mu_4-2)\delta_{ij}\delta_{jp}\delta_{pq}.
 \label{eq:pw-fourth-index-kernel}
\end{equation}
To verify \eqref{eq:pw-fourth-index-kernel}, independence and the moments
\eqref{eq:pw-symbol-moments} give the exhaustive cases
\begin{equation}
 \mathbb E[\overline{s_i}s_js_p\overline{s_q}]
 =\begin{cases}
 0,&\text{some index occurs once},\\
 \mathbb E|s_i|^2\mathbb E|s_p|^2=1,&i=j\ne p=q,\\
 \mathbb E|s_i|^2\mathbb E|s_j|^2=1,&i=p\ne j=q,\\
 \overline{\mathbb E s_i^2}\mathbb E s_j^2=0,&i=q\ne j=p,\\
 \mathbb E|s_i|^4=\mu_4,&i=j=p=q.
 \end{cases}
 \label{eq:pw-fourth-index-cases}
\end{equation}
The Kronecker terms in \eqref{eq:pw-fourth-index-kernel} have exactly
these values. In the last case they sum to \(2+(\mu_4-2)=\mu_4\).
All terms are integrable, since H\"older gives
\begin{equation}
 \mathbb E|\overline{s_i}s_js_p\overline{s_q}|
 \le\prod_{a\in(i,j,p,q)}(\mathbb E|s_a|^4)^{1/4}=\mu_4.
 \label{eq:pw-fourth-integrability}
\end{equation}
The product in \eqref{eq:pw-fourth-integrability} has four factors,
including repeated indices.  For any matrix \(\mathbf A\), finite
expansion and \eqref{eq:pw-fourth-index-kernel} give
{\begin{align}
 \mathbb E|\mathbf s^{\mathrm H}\mathbf A\mathbf s|^2
 &=\sum_{i,j,p,q}A_{ij}\overline{A_{pq}}
                     \mathbb E[\overline{s_i}s_js_p\overline{s_q}],
 \nonumber\\
 &=\sum_{i,p}A_{ii}\overline{A_{pp}}+
   \sum_{i,j}|A_{ij}|^2+(\mu_4-2)\sum_i|A_{ii}|^2,
 \nonumber\\
 &=|\operatorname{tr}\mathbf A|^2+\|\mathbf A\|_F^2
                         +(\mu_4-2)\sum_i|A_{ii}|^2.
 \label{eq:pw-quadratic-fourth-identity}
\end{align}}
Define the noncyclic shift and its unitary conjugate by
\begin{equation}
 [\mathbf J_k]_{ij}=\mathbb{I}_{\{j=i+k,\ 0\le i<N-k\}},\qquad
 \mathbf A=\mathbf U^{\mathrm H}\mathbf J_k\mathbf U.
 \label{eq:pw-shift-matrix}
\end{equation}
Unitarity {\eqref{eq:tit-unitary-modulation}} yields
{\begin{subequations}
\begin{align}
 \operatorname{tr}\mathbf A&=\operatorname{tr}\mathbf J_k=0,
 \label{eq:pw-shift-trace}\\
 \|\mathbf A\|_F^2
 &=\operatorname{tr}(\mathbf U^{\mathrm H}\mathbf J_k^{\mathrm H}
                         \mathbf J_k\mathbf U)=N-k,
 \label{eq:pw-shift-frobenius}\\
 A_{nn}&=\mathbf u_n^{\mathrm H}\mathbf J_k\mathbf u_n
 =\gamma_k(\mathbf u_n).
 \label{eq:pw-shift-diagonal}
\end{align}
\end{subequations}}

Since \(\mathbf y^{\mathrm H}\mathbf J_k\mathbf y=\gamma_k(\mathbf y)\),
substitution of \eqref{eq:pw-shift-trace}--\eqref{eq:pw-shift-diagonal}
into \eqref{eq:pw-quadratic-fourth-identity} gives
\eqref{eq:pw-eisl-single-reduction}. Summing over \(k=1,\ldots,r\)
gives \eqref{eq:pw-eisl-prefix-reduction}.

\end{IEEEproof}
\par\endgroup

\begingroup
\section{Proof of Theorem~\ref{cor:pw-eisl}}
\label{app:tit-main-theorem}
\begin{IEEEproof}
\begingroup
Equations~\eqref{eq:tit-single-gap-equivalence} and \eqref{eq:tit-prefix-gap-equivalence} translate the deterministic correlation bounds into expected sidelobe bounds, since \(2-\mu_4>0\). Applying the fourth moment identity \eqref{eq:pw-eisl-single-reduction} to Proposition~\ref{prop:pw-single-lag} gives the exact minimum \eqref{eq:tit-optimal-single-esl}. The matrix \eqref{eq:pw-block-competitor-matrix} attains it, and multiplying \eqref{eq:tit-single-correlation-gain} by \(2-\mu_4\) gives \eqref{eq:tit-single-esl-gain} and the divisibility criterion. The prefix bound \eqref{eq:pw-prefix-bound} and the \ac{IDFT} value \eqref{eq:tit-idft-benchmark} give \eqref{eq:tit-optimal-partial-eisl}, while Proposition~\ref{prop:unique} gives the necessary matrix form \eqref{eq:tit-optimal-modulation-family}. Conversely, a matrix of that form is unitary because it is obtained from \(\mathbf V\) by diagonal unitary factors and a column permutation. Equations~\eqref{eq:tit-family-correlations}--\eqref{eq:tit-family-single-energy} show that its correlation energy agrees with that of \(\mathbf V\) at every lag. Summing over \(k=1,\ldots,r\) and applying \eqref{eq:pw-eisl-prefix-reduction} proves attainment of the prefix minimum.
\par\endgroup
\end{IEEEproof}
\par\endgroup

\begingroup
\section{Proof of Corollary~\ref{cor:tit-simultaneous-prefix}}
\label{app:tit-simultaneous-prefix}
\begin{IEEEproof}
If \eqref{eq:tit-simultaneous-divisibility} holds, the single lag part of Theorem~\ref{cor:pw-eisl} shows that the \ac{IDFT} satisfies \eqref{eq:tit-simultaneous-prefix}. Conversely, a matrix satisfying \eqref{eq:tit-simultaneous-prefix} minimizes \(\operatorname{\acs*{ESL}}_1\), which equals \(\mathcal E_1\) by \eqref{eq:tit-partial-eisl}. The equality characterization in Theorem~\ref{cor:pw-eisl} with \(r=1\) therefore gives \eqref{eq:tit-optimal-modulation-family}. Equations~\eqref{eq:tit-family-correlations} and \eqref{eq:tit-family-single-energy} apply to this matrix. The fourth moment identity \eqref{eq:pw-eisl-single-reduction} then gives \(\operatorname{\acs*{ESL}}_k(\mathbf U_*)=\operatorname{\acs*{ESL}}_k(\mathbf V)\). Since \(\mathbf U_*\) minimizes each of these quantities for \(1\le k\le r\), Theorem~\ref{cor:pw-eisl} implies \eqref{eq:tit-simultaneous-divisibility}. For \(r=N-1\) and \(N\ge3\), the integer \(N-1\) does not divide \(N\), which proves the final assertion.
\end{IEEEproof}
\par\endgroup

\begingroup
\section{Proof of Lemma~\ref{lem:pw-gap}}
\label{app:tit-magnitude-identity}
\begin{IEEEproof}
Expanding the squares in \eqref{eq:pw-EH-definition} and using
\eqref{eq:pw-correlation-objective} and \eqref{eq:pw-A-definition}
gives
\begin{equation}
 E_k=A_k-2\gamma_k(\mathbf x),\qquad
 H_k=A_k+2\gamma_k(\mathbf x).
 \label{eq:pw-EH-expanded}
\end{equation}
Consequently, \eqref{eq:pw-EH-expanded} and
\(A_k=2(N-k)-b_k\) from \eqref{eq:pw-A-definition} imply
{\begin{align}
 E_kH_k
 &=A_k^2-4\gamma_k(\mathbf x)^2
 =[2(N-k)-b_k]^2-4\gamma_k(\mathbf x)^2,
 \nonumber\\
 E_kH_k-b_k^2
 &=4\left[(N-k)^2-\gamma_k(\mathbf x)^2-(N-k)b_k\right].
 \label{eq:pw-gap-cancellation}
\end{align}}
Rearranging \eqref{eq:pw-gap-cancellation} proves
\eqref{eq:pw-gap-identity}.
\end{IEEEproof}
\par\endgroup

\begingroup
\section{Proof of Lemma~\ref{lem:pw-contraction}}
\label{app:tit-contraction}
\begin{IEEEproof}
First combine \eqref{eq:pw-boundary-form} with the assumed expansion \eqref{eq:pw-expansion}:
{\begin{equation}
 |b_\ell|
 \le\sum_{k\in\mathcal K}\sum_{i=0}^{N-k-1}
          |f^{(\ell)}_{i,k}|\,|p_i-p_{i+k}|\le\sum_{k\in\mathcal K}t_{\ell k}V_k.
 \label{eq:pw-contraction-grouped}
\end{equation}}
Here the second inequality uses \eqref{eq:pw-coefficient-capacity} and
\eqref{eq:pw-V-definition}.  Weighted Cauchy--Schwarz and the row bound in
\eqref{eq:pw-coefficient-budgets} give
\begin{equation}
 b_\ell^2\le\left(\sum_k t_{\ell k}\right)
                         \sum_k t_{\ell k}V_k^2
 \le\rho\sum_k t_{\ell k}V_k^2.
 \label{eq:pw-row-cauchy}
\end{equation}
Sum \eqref{eq:pw-row-cauchy} over \(\ell\) and use the column bound in
\eqref{eq:pw-coefficient-budgets}:
\begin{align}
 \sum_\ell b_\ell^2
 &\le\rho\sum_k\left(\sum_\ell t_{\ell k}\right)V_k^2,
 \nonumber \\
 &\le\rho^2\sum_kV_k^2.
 \label{eq:pw-column-budget}
\end{align}
This completes the proof.
\end{IEEEproof}
\par\endgroup

\begingroup
\section{Proof of Lemma~\ref{lem:pw-balance-solvability}}
\label{app:tit-balance-solvability}
\begin{IEEEproof}
Let \(\Set{C}_1,\ldots,\Set{C}_s\) be the connected components of the undirected graph with incidence matrix \(\mathbf A_{\mathcal K}\). By \cite[Theorem~2.3 and Lemma~2.4]{bapat2010graphs},
\begin{equation}
 \operatorname{rank}(\mathbf A_{\mathcal K})=N-s,
 \label{eq:pw-solvability-rank}
\end{equation}
and \eqref{eq:balance_matrix} is solvable for all
\(\ell\in\mathcal K\) if and only if
\begin{equation}
 \sum_{i\in\Set{C}_a}\eta_{\ell,i}=0,
 \qquad a=1,\ldots,s,\ \ell\in\mathcal K.
 \label{eq:pw-solvability-component-balance}
\end{equation}
Since every edge length is divisible by \(d\), each component must lie in a single residue class modulo \(d\).  Since all \(d\) classes are nonempty, we have \(s\ge d\).

For necessity, let us assume that \eqref{eq:pw-solvability-component-balance} holds.  For each \(\ell\in\mathcal K\), a component \(\Set{C}_a\) is a union of complete chains with step \(\ell\).  Each chain contributes one vertex to each boundary interval in \eqref{eq:pw-boundary-coefficients}, including a chain with only one vertex.  Let \(h_\ell(\Set{C}_a)\) be the number of these chains, then
\begin{equation}
 0=\sum_{i\in\Set{C}_a}\eta_{\ell,i}
 =2h_\ell(\Set{C}_a)-\frac{2\ell|\Set{C}_a|}{N}
 \quad\Longrightarrow\quad
 \frac{\ell|\Set{C}_a|}{N}\in\mathbb Z.
 \label{eq:pw-solvability-chain-integrality}
\end{equation}
Because \(d\) is the greatest common divisor of the allowed lags, it can be represented as an integer linear combination of these lags, and thus \eqref{eq:pw-solvability-chain-integrality} gives
$$
\frac{d|\Set{C}_a|}{N}=\sum_{\ell\in\Set{K}} \alpha_\ell \frac{\ell|\Set{C}_a|}{N}\in\mathbb Z_+,
$$  
for some integers $\{\alpha_{\ell}\}$. Hence
\begin{equation}
 |\Set{C}_a|\ge\frac Nd,\qquad
 N=\sum_{a=1}^{s}|\Set{C}_a|\ge\frac{sN}{d}
 \quad\Longrightarrow\quad s\le d.
 \label{eq:pw-solvability-component-count}
\end{equation}
Together with \(s\ge d\), \eqref{eq:pw-solvability-component-count}
forces \(s=d\) and \(|\Set{C}_a|=N/d\) for every \(a\).
Thus \(d\mid N\), and \eqref{eq:pw-solvability-rank} gives the
required rank.

Conversely, suppose \(d\mid N\) and the rank is \(N-d\). By \eqref{eq:pw-solvability-rank}, \(s=d\), so each component is an entire residue class modulo \(d\).  Such a class has \(N/d\) vertices. Furthermore, since \(d\mid\ell\), each class has exactly \(\ell/d\) vertices in each boundary interval.  Therefore
\begin{equation}
 \sum_{i\in\Set{C}_a}\eta_{\ell,i}
 =\frac{2\ell}{d}-\frac{2\ell}{N}\frac Nd=0,
 \qquad a=1,\ldots,d,\ \ell\in\mathcal K.
 \label{eq:pw-solvability-sufficient-balance}
\end{equation}
Equation~\eqref{eq:pw-solvability-sufficient-balance} verifies
the criterion \eqref{eq:pw-solvability-component-balance},
which proves sufficiency.
\end{IEEEproof}
\par\endgroup

\begingroup
\section{Proof of Proposition~\ref{prop:pw-single-lag}}
\label{app:tit-single-lag}
\begingroup
\begin{IEEEproof}

\textbf{An estimate for a vector on one chain.} For every \(m\ge1\) and \(\mathbf z\in\mathbb C^m\), we first establish

\begin{equation}
 |\gamma_1(\mathbf z)|^2\le\|\mathbf z\|_2^2
 \left[\left(1-\frac1{m^2}\right)\|\mathbf z\|_2^2
       -\left(1-\frac1m\right)(|z_0|^2+|z_{m-1}|^2)\right],
 \label{eq:tit-chain-vector-bound}
\end{equation}

where \(\gamma_1(\mathbf z)=\sum_{i=0}^{m-2}\overline{z_i}z_{i+1}\), with the empty sum equal to zero when \(m=1\). Both sides of \eqref{eq:tit-chain-vector-bound} vanish when \(m=1\) or \(\mathbf z=0\). Otherwise, normalize the magnitudes by

\begin{equation}
 x_i=\frac{\sqrt m\,|z_i|}{\|\mathbf z\|_2},\qquad p_i=x_i^2,
 \qquad \sum_{i=0}^{m-1}p_i=m.
 \label{eq:tit-chain-normalization}
\end{equation}

Use \(E_1,H_1,b_1,V_1\) from \eqref{eq:pw-EH-definition}, \eqref{eq:pw-A-definition}, and \eqref{eq:pw-V-definition} with ambient dimension \(m\). The boundary expression \eqref{eq:pw-boundary-form} and telescoping give

\begin{equation}
 b_1=p_0+p_{m-1}-2
 =\sum_{i=0}^{m-2}\left(1-\frac{2(i+1)}m\right)(p_i-p_{i+1}).
 \label{eq:tit-chain-boundary-expansion}
\end{equation}

The largest absolute coefficient in \eqref{eq:tit-chain-boundary-expansion} is \(1-2/m\). The triangle inequality and \eqref{eq:pw-variation-bound} therefore imply

\begin{equation}
 b_1^2\le\left(1-\frac2m\right)^2V_1^2\le E_1H_1.
 \label{eq:tit-chain-boundary-bound}
\end{equation}

Apply Lemma~\ref{lem:pw-gap} with \(N=m\) and \(k=1\). Equations~\eqref{eq:tit-chain-boundary-expansion} and \eqref{eq:tit-chain-boundary-bound} give

\begin{align}
 \gamma_1(\mathbf x)^2
 &\le(m-1)^2-(m-1)b_1\nonumber\\
 &=m^2-1-(m-1)(p_0+p_{m-1}).
 \label{eq:tit-chain-magnitude-bound}
\end{align}

The triangle inequality, \eqref{eq:tit-chain-normalization}, and \eqref{eq:tit-chain-magnitude-bound} now yield

\begin{align}
 |\gamma_1(\mathbf z)|^2
 &\le\frac{\|\mathbf z\|_2^4}{m^2}\gamma_1(\mathbf x)^2\nonumber\\
 &\le\|\mathbf z\|_2^2
 \left[\left(1-\frac1{m^2}\right)\|\mathbf z\|_2^2
       -\left(1-\frac1m\right)(|z_0|^2+|z_{m-1}|^2)\right].
 \label{eq:tit-chain-vector-verification}
\end{align}

This proves \eqref{eq:tit-chain-vector-bound}.

\textbf{An upper bound for an arbitrary unitary matrix.} For each column \(\mathbf u_n\) of a unitary \(\mathbf U\), denote its entries on chain \(\mathcal C_j\), in chain order, by

\begin{equation}
 \mathbf z_{j,n}=(u_{j,n},u_{j+k,n},\ldots,u_{j+(m_j-1)k,n})^{\mathrm T}
 \in\mathbb C^{m_j}.
 \label{eq:tit-chain-column-restriction}
\end{equation}

The chains partition the rows. Column and row normalization in \eqref{eq:tit-unitary-modulation} therefore give

\begin{subequations}
\begin{align}
 \sum_{j=0}^{k-1}\|\mathbf z_{j,n}\|_2^2&=1,
 \label{eq:tit-chain-column-norm}\\
 \sum_{n=0}^{N-1}\|\mathbf z_{j,n}\|_2^2&=m_j,
 \label{eq:tit-chain-power-sum}\\
 \sum_{n=0}^{N-1}\left(|u_{j,n}|^2+|u_{j+(m_j-1)k,n}|^2\right)&=2.
 \label{eq:tit-chain-endpoint-sum}
\end{align}
\end{subequations}

Equation~\eqref{eq:tit-chain-endpoint-sum} also holds when \(m_j=1\), in which case the same row is counted twice. Every sample pair at lag \(k\) is adjacent within one chain. Apply Cauchy--Schwarz to the chains on which \(\mathbf z_{j,n}\ne0\), and then use \eqref{eq:tit-chain-vector-bound} and \eqref{eq:tit-chain-column-norm}, to obtain

\begin{align}
 |\gamma_k(\mathbf u_n)|^2
 &=\left|\sum_{j=0}^{k-1}\gamma_1(\mathbf z_{j,n})\right|^2\nonumber\\
 &\le\left(\sum_{j=0}^{k-1}\|\mathbf z_{j,n}\|_2^2\right)
       \sum_{\substack{0\le j<k\\\|\mathbf z_{j,n}\|_2>0}}
       \frac{|\gamma_1(\mathbf z_{j,n})|^2}{\|\mathbf z_{j,n}\|_2^2}\nonumber\\
 &\le\sum_{j=0}^{k-1}\left[
       \left(1-\frac1{m_j^2}\right)\|\mathbf z_{j,n}\|_2^2
       -\left(1-\frac1{m_j}\right)
       (|u_{j,n}|^2+|u_{j+(m_j-1)k,n}|^2)\right].
 \label{eq:tit-chain-column-bound}
\end{align}

Terms with \(\mathbf z_{j,n}=0\) contribute zero to the final sum. Summing \eqref{eq:tit-chain-column-bound} over \(n\), and using \eqref{eq:tit-chain-power-sum} and \eqref{eq:tit-chain-endpoint-sum}, gives

\begin{align}
 S_{\{k\}}(\mathbf U)
 &\le\sum_{j=0}^{k-1}\left[
       \left(1-\frac1{m_j^2}\right)m_j-2\left(1-\frac1{m_j}\right)\right]\nonumber\\
 &=\sum_{j=0}^{k-1}\frac{(m_j-1)^2}{m_j}.
 \label{eq:tit-arbitrary-single-upper-bound}
\end{align}

\textbf{Attainment and comparison with the full \ac{IDFT}.} The matrix \(\mathbf U_{\rm blk}\) in \eqref{eq:pw-block-competitor-matrix} is unitary because \(\mathbf P\) and each block \(\mathbf F_{m_j}^{\mathrm H}\) are unitary. Every column is supported on one chain. For \(m_j\ge2\), \eqref{eq:tit-idft-correlation} with \(N=m_j\) and \(k=1\) gives a correlation magnitude of \((m_j-1)/m_j\) within that chain. A chain of length one contributes zero. Using \(\sum_{j=0}^{k-1}m_j=N\) and \eqref{eq:pw-chain-lengths}, we obtain

\begin{align}
 S_{\{k\}}(\mathbf U_{\rm blk})
 &=\sum_{j=0}^{k-1}m_j\left(\frac{m_j-1}{m_j}\right)^2\nonumber\\
 &=N-2k+\sum_{j=0}^{k-1}\frac1{m_j}\nonumber\\
 &=N-2k+\frac{h}{q+1}+\frac{k-h}{q}.
 \label{eq:pw-block-sum}
\end{align}

This attains \eqref{eq:tit-arbitrary-single-upper-bound} and proves \eqref{eq:tit-single-correlation-optimum}. Finally, subtracting \eqref{eq:tit-idft-benchmark} and using \(N=qk+h\) gives

\begin{align}
 S_{\{k\}}(\mathbf U_{\rm blk})-S_{\{k\}}(\mathbf V)
 &=\frac{h}{q+1}+\frac{k-h}{q}-\frac{k^2}{N}\nonumber\\
 &=\frac{h(k-h)}{Nq(q+1)}.
 \label{eq:tit-single-gain-calculation}
\end{align}

The right side is zero exactly when \(h=0\), and is strictly positive for \(0<h<k\). This proves \eqref{eq:tit-single-correlation-gain} and the stated divisibility criterion.

\end{IEEEproof}
\par\endgroup
\par\endgroup

\begingroup
\section{Proof of Lemma~\ref{lem:pw-interval}}
\label{app:tit-path-arrangement}
\begin{IEEEproof}
If \(\ell\ge N-r-1\), then \(\Set{Q}_\ell\) is empty by \eqref{eq:pw-prefix-correction-pairs}, so all counts are zero. Otherwise, construct a path for each \((a,b)\in\Set{Q}_\ell\) as follows. If \(b\le r\), use the direct path \(\Set{P}_{ab}=(a,b)\). If \(b>r\), choose the unique integer
\begin{equation}
 v_{ab}\in\{b-r,\ldots,b-r+\ell-1\},\qquad
 v_{ab}\equiv a\pmod\ell.
 \label{eq:pw-path-last-vertex}
\end{equation}
The interval in \eqref{eq:pw-path-last-vertex} contains \(\ell\) consecutive integers, so there is a unique choice. The congruence and \(0\le a<\ell\) give \(a\le v_{ab}<b\). Choose
\begin{equation}
 \Set{P}_{ab}=(a,a+\ell,\ldots,v_{ab},b),\qquad
 r-\ell+1\le b-v_{ab}\le r.
 \label{eq:pw-q-path-choice}
\end{equation}
When \(v_{ab}=a\), this is just the direct path, otherwise all steps before the last have length \(\ell\). Since \(\ell\le r\), every path in \eqref{eq:pw-q-path-choice} satisfies \eqref{eq:pw-correction-path-choice}.

Next we compute the edge reuse times. Fix an allowed edge \((i,i+k)\). If it is the last edge of a path, we have \(b=i+k\). When \(b\le r\), it is a direct edge and hence \(a=i\). When \(b>r\), \eqref{eq:pw-path-last-vertex} determines \(a\) as the unique integer in \(\{0,\ldots,\ell-1\}\) congruent to \(i\) modulo \(\ell\). Thus at most one chosen path uses this edge as its last step, namely
\begin{equation}
 \bigl|\{(a,b)\in\Set{Q}_\ell:
 (i,i+k)\text{ is the last edge of }\Set{P}_{ab}\}\bigr|\le1.
 \label{eq:pw-last-edge-count}
\end{equation}
Every edge used before the last step has length \(\ell\). If \((i,i+\ell)\) is such an edge, its initial vertex again determines the starting vertex \(a\) modulo \(\ell\). For this fixed \(a\), there is at most one path per terminal vertex \(b\), and $b$ satisfies \(r<b<N-r-1\) by \eqref{eq:pw-correction-endpoints} because every path with \(b\le r\) has only one edge. Writing \(x_+=\max\{x,0\}\), the number of valid $b$ is at most
\begin{equation}
 \bigl|\{b\in\mathbb Z:r<b<N-r-1\}\bigr|=(N-2r-2)_+.
 \label{eq:pw-initial-step-terminals}
\end{equation}
Combining \eqref{eq:pw-last-edge-count} and \eqref{eq:pw-initial-step-terminals} gives \(c_{i,\ell}^{(\ell)}\le1+(N-2r-2)_+ =\max\{1,N-2r-1\}\). For \(k\ne\ell\), the edge can only occur as a last step, so \eqref{eq:pw-last-edge-count} gives \(c_{i,k}^{(\ell)}\le1\). This proves \eqref{eq:pw-correction-count-bound}. 

Finally, by \eqref{eq:pw-correction-endpoints}, all the vertices in every permissible path have indices less than \(N-r-1\), which proves the stated vanishing of the edge reuse times.
\end{IEEEproof}
\par\endgroup

\begingroup
\section{Proof of Proposition~\ref{prop:unique}}
\label{app:tit-prefix-equality}
\begin{IEEEproof}
Assume equality in \eqref{eq:pw-prefix-bound}. By
\eqref{eq:pw-deficit-definition}, this means
\(B_{\{1,\ldots,r\}}(\mathbf U)=0\).
We determine the absolute values of the entries of \(\mathbf U\),
then their phases.

\par\medskip\noindent
\textbf{All entries of \(\mathbf U\) have absolute value \(N^{-1/2}\).}
For any nonnegative vector \(\mathbf x\) with
\(\sum_{i=0}^{N-1}x_i^2=N\), the row and column bounds
\eqref{eq:pw-unified-prefix-budgets} allow us to apply
Lemma~\ref{lem:pw-contraction} with
\(\rho_{N,r}=\max\{r,N-2\}/N<1\).
Its conclusion \eqref{eq:pw-boundary-contraction}, followed by
\eqref{eq:pw-variation-bound}, gives
\begin{equation}
 \begin{aligned}
 \sum_{k=1}^{r}(E_kH_k-b_k^2)
 &\ge\sum_{k=1}^{r}E_kH_k
       -\rho_{N,r}^{\,2}\sum_{k=1}^{r}V_k^2\\
 &\ge(1-\rho_{N,r}^{\,2})\sum_{k=1}^{r}E_kH_k\ge0.
 \end{aligned}
 \label{eq:pw-strict-profile}
\end{equation}
Apply \eqref{eq:pw-strict-profile} to each vector
\(\mathbf x^{(n)}\) defined in \eqref{eq:pw-normalized-profile}.
In the gap decomposition \eqref{eq:pw-deficit-decomposition},
the sum over \(k\) in the first term is nonnegative for each
column \(n\), by \eqref{eq:pw-strict-profile}. Every summand in
the second term is nonnegative by \eqref{eq:pw-phase-relaxation}.
Since \(B_{\{1,\ldots,r\}}(\mathbf U)=0\), these nonnegative
quantities must all be zero. Thus, for every \(0\le n<N\),
\begin{equation}
 \begin{gathered}
 \sum_{k=1}^{r}
 \left[E_k^{(n)}H_k^{(n)}-(b_k^{(n)})^2\right]=0,~\gamma_k(\mathbf x^{(n)})^2
 =N^2|\gamma_k(\mathbf u_n)|^2,~ 1\le k\le r.
 \end{gathered}
 \label{eq:pw-zero-gap-conditions}
\end{equation}
Because \(1-\rho_{N,r}^{\,2}>0\), equations
\eqref{eq:pw-strict-profile} and \eqref{eq:pw-zero-gap-conditions}
imply \(E_1^{(n)}H_1^{(n)}=0\).
On the other hand, \eqref{eq:pw-EH-definition} and \eqref{eq:pw-profile-normalization} give $ H_1^{(n)} =\sum_{i=0}^{N-2}(x_i^{(n)}+x_{i+1}^{(n)})^2>0$. It follows that
\begin{equation}
 E_1^{(n)}=\sum_{i=0}^{N-2}(x_i^{(n)}-x_{i+1}^{(n)})^2=0.
 \label{eq:pw-adjacent-magnitudes-equal}
\end{equation}
Equation~\eqref{eq:pw-adjacent-magnitudes-equal} makes all adjacent
entries of \(\mathbf x^{(n)}\) equal. Their common value is \(1\)
by \eqref{eq:pw-profile-normalization}. Therefore,
\begin{equation}
 x_i^{(n)}=1,\qquad |u_{i,n}|=N^{-1/2},
 \qquad 0\le i,n<N.
 \label{eq:pw-basis-equality-magnitudes}
\end{equation}

\par\medskip\noindent
\textbf{Within each column, successive entries have a constant phase difference.}
Taking \(k=1\) in \eqref{eq:pw-zero-gap-conditions} and using
\eqref{eq:pw-basis-equality-magnitudes}, we obtain
\begin{equation}
 \left|\sum_{q=0}^{N-2}\overline{u_{q,n}}u_{q+1,n}\right|
 =\frac{N-1}{N}
 =\sum_{q=0}^{N-2}|\overline{u_{q,n}}u_{q+1,n}|.
 \label{eq:pw-adjacent-triangle-equality}
\end{equation}
Equality in the triangle inequality in
\eqref{eq:pw-adjacent-triangle-equality} means that all products
\(\overline{u_{q,n}}u_{q+1,n}\) have the same argument.
Their absolute values are \(1/N\) by
\eqref{eq:pw-basis-equality-magnitudes}. Hence, for some real \(\theta_n\),
\begin{equation}
 \overline{u_{q,n}}u_{q+1,n}=N^{-1}e^{\mathrm i\theta_n}
\Longrightarrow
 u_{q+1,n}=e^{\mathrm i\theta_n}u_{q,n},\qquad
 0\le q<N-1.
 \label{eq:pw-successive-entry-ratio}
\end{equation}
Write \(u_{0,n}=N^{-1/2}e^{\mathrm i\alpha_n}\) with \(\alpha_n\in\mathbb R\). Repeatedly applying \eqref{eq:pw-successive-entry-ratio} gives
\begin{equation}
 u_{q,n}=N^{-1/2}e^{\mathrm i\alpha_n+\mathrm i q\theta_n},
 \qquad 0\le q,n<N.
 \label{eq:pw-geometric-columns}
\end{equation}

\par\medskip\noindent
\textbf{Orthogonality determines the possible values of \(\theta_n\).}
For distinct columns \(m\ne n\), equations
{\eqref{eq:tit-unitary-modulation}} and \eqref{eq:pw-geometric-columns} give
\begin{equation}
 0=\mathbf u_m^{\mathrm H}\mathbf u_n
 =\frac{e^{\mathrm i(\alpha_n-\alpha_m)}}N
   \sum_{q=0}^{N-1}e^{\mathrm i q(\theta_n-\theta_m)}.
 \label{eq:pw-column-orthogonality}
\end{equation}
If \(e^{\mathrm i(\theta_n-\theta_m)}=1\), the sum in \eqref{eq:pw-column-orthogonality} would be \(N\), causing a contradiction. The formula for a finite geometric sum therefore yields
\begin{equation}
 0=\frac{1-e^{\mathrm i N(\theta_n-\theta_m)}}
        {1-e^{\mathrm i(\theta_n-\theta_m)}}
 \quad\Longrightarrow\quad
 e^{\mathrm i N(\theta_n-\theta_m)}=1.
 \label{eq:pw-phase-difference-roots}
\end{equation}
Equation~\eqref{eq:pw-phase-difference-roots} shows that, modulo \(2\pi\), each \(\theta_n\) is one of \(\theta_0,\theta_0+2\pi/N,\ldots,\theta_0+2\pi(N-1)/N\). The values \(\theta_n\) are distinct modulo \(2\pi\), because equal values would contradict \eqref{eq:pw-column-orthogonality}. Thus each of these \(N\) possibilities occurs once, and there is a permutation \(\sigma\) of \(\{0,\ldots,N-1\}\) such that
\begin{equation}
 \theta_n\equiv\theta_0+\frac{2\pi\sigma(n)}N\pmod{2\pi}.
 \label{eq:pw-common-phase-values}
\end{equation}
\begingroup
Combining \eqref{eq:pw-geometric-columns} and \eqref{eq:pw-common-phase-values} gives
\begin{equation}
 u_{q,n}=\frac{e^{\mathrm i(\alpha_n+q\theta_0)}}{\sqrt N}
             e^{2\pi\mathrm i q\sigma(n)/N},\qquad 0\le q,n<N.
 \label{eq:tit-equality-matrix-form}
\end{equation}
Equation~\eqref{eq:tit-equality-matrix-form} is precisely \eqref{eq:tit-optimal-modulation-family}, with \(\theta=\theta_0\).
\par\endgroup

\end{IEEEproof}
\par\endgroup

\end{document}